\documentclass[journal]{IEEEtran} \onecolumn 

\usepackage{graphicx}

\usepackage{amsfonts}       

\usepackage[margin=0.7in]{geometry}

\usepackage{amsmath,amssymb,amsthm, bm, bbm}
\usepackage{algorithm}
\usepackage{array}
\usepackage{mdwmath}
\usepackage{url}
\usepackage{multirow}
\usepackage{bbm}

\usepackage{graphicx}
\usepackage{caption}
\usepackage{subcaption}

\usepackage{booktabs}

\usepackage{tikz}
\usepackage{algpseudocode}
\usepackage{pgfplots}
\newtheorem{theorem}{Theorem}[section]
\newtheorem{lemma}[theorem]{Lemma}

\newtheorem{assumption}[theorem]{Assumption}
\renewcommand\thetheorem{\arabic{section}.\arabic{theorem}}

\newcommand{\bproof}{ \begin{IEEEproof} }
	\newcommand{\eproof}{ \end{IEEEproof} }
\newcommand{\beqno}{ \begin{equation*} }
\newcommand{\eeqno}{ \end{equation*} }
\newcommand{\beqa}{\begin{eqnarray*} }
	\newcommand{\eeqa}{\end{eqnarray*} }
\newcommand{\beq}{ \begin{equation} }
\newcommand{\eeq}{ \end{equation} }
\newcommand{\ik}{{ik}}
\newcommand{\p}{\bm{p}}

\newcommand{\tot}{\mathrm{tot}}
\renewcommand{\a}{\bm{a}}

\newcommand{\x}{\bm{x}}
\newcommand{\bnu}{\bm{\nu}}
\newcommand{\e}{\bm{e}}
\newcommand{\h}{\bm{h}}
\renewcommand{\b}{\bm{b}}
\newcommand{\y}{\bm{y}}
\newcommand{\w}{\bm{w}}
\renewcommand{\O}{{\bm{O}}}

\newcommand{\U}{{\bm{U}}}
\newcommand{\V}{{\bm{V}}}
\newcommand{\X}{\bm{X}}
\newcommand{\R}{{\bm{R}}}
\newcommand{\B}{\bm{B}}
\newcommand{\G}{\bm{G}}

\newcommand{\I}{\bm{I}}
\newcommand{\Y}{\bm{Y}}
\newcommand{\g}{\bm{g}}

\newcommand{\vv}{\bm{v}}

\newcommand{\xhat}{\bm{\hat{x}}}
\newcommand{\bhat}{\bm{\hat{b}}}
\newcommand{\Uhat}{\hat\U}
\newcommand{\Bhat}{\hat\B}
\newcommand{\Xhat}{\hat\X}

\newcommand{\Z}{{\bm{Z}}}

\newcommand{\Q}{{\bm{Q}}}
\newcommand{\F}{{\bm{F}}}

\newcommand{\n}{\mathcal{N}}
\newcommand{\E}{\mathbb{E}}

\newcommand{\z}{\bm{z}}
\newcommand{\indic}{\mathbbm{1}}

\newcommand{\SE}{\mathrm{SubsDist}}
\newcommand{\dist}{\mathrm{Dist}}

\newcommand{\A}{\bm{A}}

\newcommand{\Chat}{\bm{\hat{C}}}

\newcommand{\M}{\bm{M}}

\newcommand{\Span}{\mathrm{range}}
\newcommand{\trace}{\mathrm{trace}}

\renewcommand{\S}{\mathcal{S}}
\newcommand{\W}{\bm{W}}
\renewcommand{\Re}{\mathbb{R}}

\newtheorem{claim}[theorem]{Claim}
\renewcommand\thetheorem{\arabic{section}.\arabic{theorem}}

\newcommand{\tB}{\tilde\B^*}
\newcommand{\tb}{\tilde\b^*}

\newcommand{\init}{{\mathrm{init}}}

\newcommand{\matdist}{\mathrm{MatDist}}
\newcommand{\Ustar}{\U^*{}}
\newcommand{\Vstar}{\V^*{}}
\newcommand{\Xstar}{{\X^*}}
\newcommand{\xstar}{\x^*}
\newcommand{\deltinit}{\delta_0}

\newcommand{\bSigma}{{\bm\Sigma^*}}

\newcommand{\sigmin}{{\sigma_{\min}^*}}
\newcommand{\sigmax}{{\sigma_{\max}^*}}

\newcommand{\bea}{\begin{eqnarray}} 
\newcommand{\eea}{\end{eqnarray}}

\newcommand{\Bstar}{\B^*}
\newcommand{\qreq}{\overset{\mathrm{QR}}=} 
\newcommand{\cb}{\bm{c}}
\renewcommand{\P}{\bm{P}}

\newcommand{\bstar}{\b^*}
\newcommand{\ben}{\begin{enumerate}}
	\newcommand{\een}{\end{enumerate}}

\newcommand{\svdeq}{\overset{\mathrm{SVD}}=} 

\newcommand{\bi}{\begin{itemize}} 
	\newcommand{\ei}{\end{itemize}}

\newcommand{\checkU}{\U} \newcommand{\checktB}{\B} \newcommand{\checktb}{\b}

\newcommand{\PR}{\text{\scriptsize{PR}}}

\newcommand{\SEF}{\SE_F}
\newcommand{\deltatfrob}{\delta_{t,F}}
\newcommand{\deltainitfrob}{\delta_{0,F}}
\renewcommand{\P}{\bm{P}}

\renewcommand{\o}{\bm{o}}
\renewcommand{\forall}{\text{ for all }}

\begin{document}

\title{Sample-Efficient Low Rank Phase Retrieval}
\author{%
Seyedehsara Nayer and Namrata Vaswani \\
    Iowa State University, Ames, IA, USA \\
Email: \texttt{\{sarana,namrata\}@iastate.edu}
}

\maketitle


%
\begin{abstract}
This work studies the Low Rank Phase Retrieval (LRPR) problem: recover an $n \times q$ rank-$r$ matrix $\Xstar$ from $\y_k = |\A_k^\top \xstar_k|$,  $k=1, 2,..., q$, when each $\y_k$ is an m-length vector containing independent phaseless linear projections of $\xstar_k$. The different matrices $\A_k$ are i.i.d. and each contains i.i.d. standard Gaussian entries. We obtain an improved guarantee for AltMinLowRaP, which is an Alternating Minimization solution to LRPR that was introduced and studied in our recent work. As long as the right singular vectors of $\Xstar$ satisfy the incoherence assumption, we can show that the AltMinLowRaP estimate converges geometrically to  $\Xstar$ if the total number of measurements $mq \gtrsim nr^2 (r + \log(1/\epsilon))$. In addition, we also need $m \gtrsim  max(r, \log q, \log n)$ because of the specific asymmetric nature of our problem. Compared to our recent work, we improve the sample complexity of the AltMin iterations by a factor of $r^2$, and that of the initialization by a factor of $r$. We also extend our result to the noisy case; we prove stability to corruption by small additive noise.
\end{abstract}


\section{Introduction}
The generalized phase retrieval (PR) problem -- recover an $n$-length signal $\xstar$ from measurements $\y:=| \A^\top \xstar|$ where $\A$ is a known $n \times m$ matrix -- has been extensively studied in the last decade \cite{candes_phaselift, pr_altmin,wf}. Here $^\top$ denotes (conjugate) transpose and $|.|$ denotes element-wise magnitudes.
PR is a classical problem that occurs in many applications such as X-ray crystallography, astronomy, and ptychography because the phase information is either difficult or impossible to obtain. 
Recent works \cite{twf,altmin_irene_w,rwf,taf} have developed provably correct and fast recovery algorithms for PR that can both achieve order-optimal sample complexity (require $m \ge C n$); and work in near linear time, $C mn \log(1/\epsilon)$, when $\A$ contains independent identically distributed (i.i.d.) standard Gaussian entries; \cite{twf,rwf,taf} assume real-valued Gaussian entries in $\A$, while \cite{altmin_irene_w} assumes complex-valued Gaussians. Here and below $C$ denotes a different numerical constant in each use.

The only way to reduce the sample complexity $m$ to less than $n$ is by imposing assumptions on $\xstar$. Sparsity is a commonly used assumption. The sparse PR problem (recover an $s$-sparse signal $\xstar$ from $\y:=|\A^\top \xstar|$) has received significant attention in recent years \cite{voroninski13,jaganathan2013sparse2,pr_altmin,sparta,cai,fastphase}. Low rank is another common assumption. As explained in \cite{lrpr_tsp, lrpr_icml,lrpr_it}, the practical way to impose it is to consider joint recovery of a set $q$ of correlated signals, that together form an (exactly or approximately) low rank matrix, from $m$ different phaseless linear projections of each of the $q$ signals. 
This model, dubbed {\em Low Rank Phase Retrieval (LRPR)} \cite{lrpr_tsp, lrpr_icml,lrpr_it}, is useful to enable fast and low-cost dynamic phaseless imaging applications, such as dynamic Fourier ptychography, where measurement acquisition is slow or expensive \cite{TCIgauri}. 

\subsection{The LRPR Problem and Notation}

\subsubsection{The LRPR Problem}
 LRPR involves recovering an $n \times q$ rank-$r$ matrix $\Xstar$  from 
\bea
\y_k : = |\A_k{}^\top \xstar_k|, \ k \in [q],
\label{obsmod}
\eea
when the $\A_k$s are $n \times m$  i.i.d. matrices with each containing i.i.d. (real- or complex-valued) standard Gaussian entries. 
Here, $[q]:= \{1,2,\dots, q\}$ and $\xstar_k$ is the $k$-th column of $\Xstar$.
Thus each scalar measurement $\y_\ik$ satisfies
\[
\y_\ik : = | \langle \a_\ik, \xstar_k \rangle |,   \ i \in [m], \ k \in [q].
\]
The above problem with $|.|$ removed is commonly referred to as {\em ``compressive PCA" or ``PCA via random projections".} Since it is the linear version of LRPR, we often refer to it as {\em linear LRPR} in this paper.

Observe that our measurements are not global, i.e., no $\y_\ik$ is a function of the entire matrix $\Xstar$. They are global for each column ($\y_\ik$ is a function of column $\xstar_k$), but not across the different columns.
 We thus need an assumption that enables correct interpolation across the different columns. The following incoherence (w.r.t. the canonical basis) assumption on the right singular vectors suffices for this purpose \cite{lrpr_it}. This type of assumption on both the left and the right singular vectors was originally introduced to make the low rank matrix completion (LRMC) problem well posed \cite{matcomp_candes,optspace,lowrank_altmin}. 

Let
\[
\Xstar \svdeq \Ustar \underbrace{\bSigma \overbrace{\Vstar{}^\top}^{\Bstar}}_{\tB}
\]
denote its reduced (rank $r$) SVD and  $\kappa:= \sigmax/\sigmin$ denote the condition number of $\bSigma$ which is  $r \times r$. Here $\Ustar$ and $\Vstar$ are tall matrices with orthonormal columns {\em (basis matrices)},  $\Ustar$ is  $n \times r$ and $\Vstar$ is $q \times r$. We let $\Bstar:=\Vstar{}^\top$ and $\tB:= \bSigma \Vstar{}^\top$ (is $r \times q$), these definitions makes it simpler to explain our recovery algorithms (our problem is such that each column of $\tB$ needs to be recovered individually in the AltMin approach).

\begin{assumption}[Right singular vectors' incoherence]
We assume that  $$\max_k \|\bstar_k\|_2 \le \mu \sqrt{r/q}$$ for a constant $\mu \ge 1 $ ($\mu$ does not grow with  $n,q,r$)\footnote{Notice that $\bstar_k = \Vstar^\top \e_k$ and thus we are imposing a bound of $\mu \sqrt{r/q}$ on row norms of the matrix of right singular vectors, $\Vstar$}.
This implies that
$
\max_k \|\xstar_k\|_2 = \max_k \|\tb_k\|_2 \le \sigmax \mu \sqrt{r/q}. 
$
This further implies that $\max_k \|\xstar_k\|_2 \le \kappa \mu {\|\Xstar\|_F}/{\sqrt{q}}$.
\label{right_incoh}
\end{assumption}

\subsubsection{Notation}
Everywhere, $\|.\|_F$ denotes the Frobenius norm,  $\|.\|_2$ or just $\|.\|$ denotes the (induced) $l_2$ norm, and $^\top$ denotes (conjugate) transpose. We use $\e_k$ to denote the $k$-th canonical basis vector ($k$-th column of $\I$).
For a complex number, $z$, $\bar{z}$ denotes the complex conjugate and we (mis)use the term ``phase" to refer to $\mathrm{phase}(z):= z/|z|$ as also done in earlier works on PR, e.g., \cite{pr_altmin}.
An $n$-length vector $\a$ is a real-valued standard Gaussian if $\a \sim \n(0,\I)$ (the entries are zero mean, unit variance and mutually independent). An $n$-length vector $\a$ is a complex-valued standard Gaussian if $\a = \a_{real} + j \a_{imag}$ with $\a_{real},\a_{imag}$ being mutually independent, and $\a_{real} \sim \n(0, 0.5 \I)$ and $\a_{imag} \sim \n(0, 0.5\I)$.

We use $\dist(\xstar,\xhat): = \min_{\theta \in [-\pi, \pi]} \|\xstar - e^{-j \theta} \xhat\|$ to denote the phase invariant distance between two vectors and we define $\matdist(\Xstar,\Xhat)^2:= \sum_{k=1}^q \dist(\xstar_k, \xhat_k)^2$.
For real-valued data, $\dist(\xstar,\xhat) = \min( \|\xstar - \xhat\|, \|\xstar + \xhat\|)$.
We say $\hat\X$ is an estimate of $\Xstar$ with $\epsilon$ accuracy if $\matdist(\Xstar,\Xhat) \leq \epsilon \|\Xstar\|_F$. {\em Without loss of generality, as done in past works on PR  \cite{pr_altmin,lrpr_it}, at some places, we assume that $\xstar_k$ is replaced by $\bar{z} \xstar_k$ where $z = \mathrm{phase}(\langle \xstar_k, \xhat_k \rangle)$. With this,  $\dist(\xstar_k, \xhat_k) = \|\xstar_k - \xhat_k\|$. A similar replacement can be done for each column of $\Xhat$ as well so that $\matdist(\Xstar,\Xhat) = \|\Xstar - \Xhat\|_F$.}

To quantify the distances between $r$-dimensional subspaces of $\Re^n$ or $\mathbb{C}^n$, represented by their $n \times r$ basis matrices $\U_1,\U_2$ (matrices with orthonormal columns), we use one of the two following metrics:
$$
\SE_2(\U_1,\U_2) :=  \|(\I - \U_1 \U_1^\top) \U_2\| \text{ and }  \SEF(\U_1, \U_2): = \|(\I - \U_1 \U_1{}^\top)\U_2\|_F.
$$
The former measures the sine of the maximum principal angle between the two subspaces while the latter measures the $l_2$ norm of the sines of all the $r$ principal angles. As a result,
\[
\SE_2(\U_1,\U_2) \le \SEF(\U_1, \U_2) \le \sqrt{r} \SE_2(\U_1, \U_2)
\]
Both are distances and hence symmetric, i.e., $\SE(\U_1, \U_2) = \SE(\U_2, \U_1)$.%
%

We use $\sum_\ik$ as short for the double summation $\sum_{k=1}^q \sum_{i=1}^m$ and $\sum_k $  for $\sum_{k=1}^q$. Letters $C,c$ are reused to denote different numerical constants in each use, with $C \ge 1$ and $c < 1$.
The notations $a \gtrsim b$ and ``$a$ is $\Omega(b)$" both mean $a \ge C b$. Finally, if in a discussion, we say that we ignore dependence on $\kappa,\mu$, it means that we are treating these as numerical constants and only considering dependence on $n,q,r$ for simplicity.

\subsection{Our Contributions, Existing Work, and Differences from Related Problems}

\subsubsection{Our Contributions}
We obtain a significantly improved guarantee for the AltMinLowRaP algorithm from our recent work \cite{lrpr_it} and also study its stability to additive noise in the measurements.  AltMinLowRaP is an alternating minimization (AltMin) based non-convex algorithm\footnote{Direct iterative algorithm that does not solve a convex relaxation} with a time complexity of $mqnr \log^2(1/\epsilon)$.
We show that, if Assumption \ref{right_incoh} (right singular vectors' incoherence) holds, if $mq \ge C_{\kappa,\mu}  \cdot  n r^2 (r + \log (1/\epsilon))$ with $C_{\kappa,\mu} = C \kappa^8 \mu^2$ and $m \ge C \max(r,\log q, \log n)$, then with high probability (w.h.p.), we can recover $\Xstar$ to $\epsilon$ accuracy in at most order $\log(1/\epsilon)$ iterations (geometric convergence).  We argue, based on comparison with the two most related well-studied problems -- sparse PR (global but phaseless measurements) and LRMC (non-global but linear (with-phase) measurements) -- why the sample complexity cannot be improved any further for any non-convex solution to LRPR. LRPR is a problem with both non-global and phaseless measurements.
We  extend our result to also handle complex-valued Gaussian measurements. Finally, we provide a stability guarantee as well, i.e., we obtain an error for bound for the noisy version of LRPR.

\subsubsection{Existing Work}
LRPR was first studied in \cite{lrpr_tsp} where we introduced an alternating minimization (AltMin) algorithm and analyzed its initialization step. In recent work \cite{lrpr_it} (and its conference version \cite{lrpr_icml}), we developed the first provable solution for it that we called AltMinLowRaP (AltMin for Low Rank PR). We also showed extensive numerical experiments that demonstrated the practical power of AltMinLowRaP.
%
Our guarantee from \cite{lrpr_it} needed  $mq \ge C_{\kappa,\mu}  n r^4  \log (1/\epsilon)$, and $m \ge C \max(r,\log q, \log n)$. Here $C_{\kappa,\mu}= C \kappa^{10} \mu^4$. Our current result improves the AltMin iterations' sample complexity by a factor of $r^2$ and that of the initialization step by a factor of $r$. 

Even the linear version of LRPR, PCA via random projections or compressive PCA, has received little attention until recently. There have been some older attempts to develop a solution and try to analyze sub-parts of it \cite{hughes_icip_2012,hughes_icml_2014, aarti_singh_subs_learn}. Since linear LRPR is a special case of LRPR, AltMinLowRaP \cite{lrpr_icml,lrpr_it} provably solves this problem as well. In more recent work \cite{lee2019neurips}, a provable convex relaxation was introduced. In our notation, this needs a slightly different version of Assumption \ref{right_incoh} and  $mq > \frac{r (n+q) \log^6 (n+q) }{ \epsilon^2}$.  The time complexity of the solver for the convex program is not discussed. However, it is well known that solvers for convex programs are slow compared to direct iterative (non-convex) algorithms: they either require number of iterations proportional to $1/\sqrt{\epsilon}$ or have cubic dependence on the problem size (here $(nr)^3$), e.g., see \cite{lowrank_altmin,pr_altmin,lowrank_altmin_no_kappa,fastmc} and references therein.
The comparison with our current result is as follows. (i) AltMinLowRaP is significantly faster: its time complexity depends poly-logarithmically on $1/\epsilon$ and linearly on $mqnr$. 
(ii) The sample complexity of \cite{lee2019neurips} has near optimal dependence on $n,q,r$, but not on $\epsilon$. Overall, our sample complexity is better than theirs whenever the desired accuracy level $\epsilon < 1/r$. (iii) The main focus of  \cite{lee2019neurips} was to obtain the best possible result for the noisy measurements' setting. But, because of this, their noise-free case sample complexity depends on $1/\epsilon^2$. Also, their noisy case result assumes a statistical model on the noise, it is modeled as zero mean i.i.d. Gaussian and independent of the data.
On the other hand, our goal is to obtain the best sample complexity guarantee for the noise-free case and then study the stability to noise corruption under the same assumptions. We do not make any statistical assumptions on the noise,  but we need a lower bound on signal-to-noise ratio (SNR). In this sense, our noisy case result is similar to guarantees for non-convex solutions to other problems that are proved under similar assumptions, e.g., \cite{twf,rwf}. 

The previous results for LRPR \cite{lrpr_tsp,lrpr_it} as well as for its linear version \cite{lee2019neurips} all need to assume right singular vectors' incoherence (Assumption \ref{right_incoh}) or a stronger version of it.
We used this assumption in \cite{lrpr_icml,lrpr_it}. In the first work on LRPR \cite{lrpr_tsp}, we used a slightly stronger version of it (each entry of $\tb_k$ was assumed to be bounded). The authors of \cite{lee2019neurips} assumed $\max_k \|\xstar_k\|^2 \le C_0 /\|\Xstar\|_F^2/q$ and their error bound depended upon $C_0$. This assumption implies $\|\b_k\|^2 \le C_0 \kappa^2 (r/q)$, i.e. that Assumption \ref{right_incoh} holds with $\mu^2 \equiv C_0 \kappa^2$; and it is implied by Assumption \ref{right_incoh} with $C_0 \equiv \kappa^2 \mu^2$.




\subsubsection{Related Problems}
The multivariate regression (MVR) problem, studied in \cite{wainwright_linear_columnwise}, is the linear version of LRPR with $\A_k = \A$, i.e., the same $\A$ is used for different columns.  
With $\A_k =\A$, the different $\y_k$s are no longer independent of each other.
Thus, in case of MVR, the authors cannot exploit law of large numbers' arguments over all $mq$ scalar measurements $\y_\ik$. 
Consequently, the required value of $m$ for MVR can never be less than $n$. The result of  \cite{wainwright_linear_columnwise} shows that $m$ of order $(n+q)r$ is both necessary and sufficient.
On the other hand, our $mq$ scalar measurements $\y_\ik$ are all mutually independent. 
A typical error term that needs to be bounded in our case consists of a summation over $mq$ terms, with each summand depending on one $\y_\ik$. The $\y_\ik$s are not identically distributed for different $k$ but, by using the right singular vectors' incoherence assumption, we can argue that the distributions are similar enough so that concentration inequalities can be applied jointly for all the $mq$ terms.
This is what makes it possible to prove guarantees that need $m \ll n$ in our case. 

LRPR involves recovery from phaseless measurements that only depend on individual columns of $\Xstar$ and not on the entire $\Xstar$. This non-global measurement setting is what makes LRPR a more difficult problem than Sparse PR for which each $\y_i$ is a function of the entire sparse signal $\xstar$. Besides the missing phase, this is the also main reason why it is more difficult than low rank matrix sensing (LRMS) \cite{lowrank_altmin}: recover $\Xstar$ from $\y_i = \langle \A_i, \Xstar\rangle$, $i=1,2,\dots, mq$. 

Low rank matrix completion (LRMC) -- recover $\Xstar$ from measurements of a subset of its entries --  is the most closely related linear setting to LRPR that is well-studied. It involves recovery from row-wise and column-wise local measurements, while LRPR measurements are row-wise local but column-wise global.
In order to allow for correct interpolation across rows and columns, LRMC needs an incoherence (w.r.t. the canonical basis) assumption on its left and right singular vectors, and it requires the set of observed entries to be spread out, e.g., the model assumed in most works is that each entry is observed with probability $\rho$ independent of all others (i.i.d. Bernoulli($\rho$) model) \cite{matcomp_candes,lowrank_altmin}.
Because of this, a typical error term for LRMC is a weighted sum of i.i.d. Bernoulli random variables, with each weight depending on only one matrix entry. It can thus be analyzed using matrix Bernstein \cite{tail_bound} and its extensions \cite{spectral_init_review}.
%
Since our measurements are global for each column, we need the incoherence assumption on only the right singular vectors. 
Unlike LRMC, (a) our measurement model is not symmetric across rows and columns, and (b) our measurements are not bounded. A typical error term in our case is a sum of $mq$ independent sub-exponential random variables. We have to use the sub-exponential Bernstein inequality \cite{versh_book} to bound it. In order to apply this to get the desired sample complexity lower bounds, we need algorithms and corresponding proof techniques that enable us to obtain a tight enough bound on the maximum sub-exponential norm (maximum over the $mq$ summands) in the error term\footnote{At iteration $t+1$, a bound of $\delta_t (r/q)$ on the maximum sub-exponential norm is needed; here $\delta_t$ is the subspace error bound at iteration $t$.}.
%
For example, the precursor to AltMinLowRaP introduced in \cite{lrpr_tsp} could not be analyzed because of this. 
For the same reason, none of the projected gradient descent (GD) approaches for LRMC can be directly modified to work for LRPR either, see Sec. \ref{discuss_timecomp}. 

%

A lower bound on  sample complexity is derived in \cite{struct_pr_lower_bnd} for structured PR problems that have global measurements. It thus does not apply to our setting where the measurements are not global in the matrix $\Xstar$.
A compression coding idea is used to solve the standard PR problem for compressible signals in \cite{coper_it}. Standard PR again involves global and i.i.d. measurements. This approach also cannot directly apply to our problem for a similar reason. It is an interesting open question though whether either of these works can be extended for our non-global LRPR setting. 

\subsection{Organization} 
In Sec. \ref{mainres}, we briefly explain the AltMinLowRaP algorithm, provide our new guarantee for it (for the noise-free setting), Theorem \ref{best_thm}, followed by a detailed discussion of the result, why the $nr^2$ and $nr^3$ factors are needed for the AltMin iterations and the initialization respectively, the key changes to the proof techniques that help us get a significantly improved result, and why the design of a projected GD solution is not easy and done in a different parallel work. In Sec. \ref{proof_best_thm}, we provide the main lemmas needed to prove Theorem \ref{best_thm}, along with the key ideas used to prove these lemmas, and we prove the result. The lemmas are proved in Appendix \ref{lemmas_proofs} for the real-valued measurements' setting and in Appendix \ref{proofs_complex} for the complex case. We present and discuss the noisy case guarantee in Sec. \ref{noisy_lrpr}. This is proved in Appendix \ref{sec:stability}.  Finally, we conclude in Sec. \ref{conclude}. To keep this paper compact, and since the algorithm has not changed from \cite{lrpr_it}, we do not show any new simulations here.

\begin{algorithm}[t!]
	\caption{AltMinLowRaP (AltMin for Low Rank PR). $\mathrm{PR}$ refers to the algorithm used for solving the standard PR problem. 
}
		\label{lrpr_th}
		\begin{algorithmic}[1]
			\State Parameters: $T$, $T_{\PR,t}$, $\omega$.
			\State  Partition the $m_\tot$ measurements for each $\xstar_k$ into one set of $m_0$ measurements for initialization and $2T$ disjoint sets of $m_1$ measurements for the AltMin iterations. Denote these by $\y_k^{(\tau)}, \A_k^{(\tau)}, \tau=0,1,\dots 2T$.
			\State Set $\hat{r}$ as the largest $j$ for which $\lambda_j(\Y_U) - \lambda_n(\Y_U) \ge \omega$,
\[
			\Y_U= \frac{1}{mq} \sum_{k=1}^q \sum_{i=1}^m \y_{ik}^2 \a_{ik} \a_{ik}^\top \indic_{ \left\{ \y_{ik}^2 \leq 9\kappa^2 \mu^2\frac{1}{mq}\sum_{ik} \y_{ik}^2 \right\}  }
			\]
and $\y_\ik \equiv \y_\ik^{(0)}$, $\a_\ik \equiv \a_\ik^{(0)}$.
			\State $\U^0 \gets \Uhat^0 \gets$ top $\hat{r}$ singular vectors of $\Y_U$.   
			\For{$ t = 0: T$}
			\State $\bhat_k^t \gets \mathrm{PR}( \y_k^{(t)}, (\U^{t}){}^\top \A_k^{(t)}, T_{\PR,t})$, $k \in [q]$ 
			
			\State   $\xhat_k^t \gets \U^t \bhat_k^t$, $k \in [q] $. 
			\State $\hat\cb_\ik \gets \mathrm{phase} ( \a_\ik^{(T+t)}{}^\top\xhat_k^t )$, $i \in [m], \ k \in [q]$. 
			\State Get $\B^t$ by QR decomp: $\hat{\B}^t \qreq \R_B^t \B^t $. 
			\State  {\small $\Uhat^{t+1} \leftarrow \arg\min_{\tilde\U} \sum_{k =1}^q\| \Chat_k \y_k^{(T+t)} - \A_k^{(T+t)}{}^\top \tilde\U \b_{k}^{t}\|^2$.}
			\State Get $\U^{t+1}$ by QR decomp: $\Uhat^{t+1} \qreq \U^{t+1} \R_U^{t+1} $ . \hspace{5cm}
			\EndFor
		\end{algorithmic}
\end{algorithm}

\section{The AltMinLowRaP algorithm and guarantee for noise-free LRPR}\label{mainres}



\subsection{AltMinLowRaP: AltMin for Low Rank Phase retrieval}
We study the AltMinLowRaP algorithm from \cite{lrpr_it}. It is summarized in Algorithm \ref{lrpr_th}.
AltMinLowRaP can be understood as truncated spectral initialization  (line 3, 4), followed by
minimizing
$
\sum_{k=1}^q \| \  \y_k -  |\A_k{}^\top \checkU \checktb_k| \ \|^2  
$
alternatively over $\checkU, \checktB$ with the constraint that $\checkU$ is a basis matrix.
Each of the two minimizations involves recovery from phaseless measurements, but the two problems are quite different (as explained in detail in \cite{lrpr_it}).
A simpler way to understand the approach is to split it into a three-way AltMin problem over $\Ustar$, $\tb_k$, and $\cb^*_\ik:=\mathrm{phase}(\a_\ik{}^\top \xstar_k)$.  This discussion assumes ``sample-splitting" (line 2), i.e., a new independent set of samples is used in each iteration  and for each new update of $\U$ and $\tB$. The initialization uses $m_0$ samples, the iterations use $m_1$ samples per iteration.

(1) At each new iteration, given an estimate of $\Span(\Ustar)$, denoted $\U$, we recover the $\tb_k$s, by solving easy individual $r$-dimensional noisy standard PR problems (line 6). Let $\g_k:= \U^\top \xstar_k = \U^\top \Ustar \tb_k$.
We can rewrite $\y_\ik$ as
\[
\y_\ik = | \tilde{\a}_\ik{}^\top \g_k + \a_\ik{} ^\top(\I - \U \U{}^\top)\xstar_k| : =  | \tilde{\a}_\ik{}^\top \g_k | + \bm\nu_\ik
\]
where $\tilde\a_\ik := \U^\top\a_\ik$. Due to sample-splitting, $\U$ is independent of $\a_\ik$s and so $\tilde\a_\ik$s are still i.i.d. standard Gaussian.
Thus, recovering $\g_k$ from $\y_\ik, i=1,2,\dots, m$ is an $r$-dimensional noisy PR problem with noise $\bm\nu_\ik$ satisfying $|\bm\nu_\ik| \le |\a_\ik{} ^\top(\I - \U \U{}^\top)\xstar_k|$. Using the sub-exponential Bernstein inequality  \cite{versh_book}, it can be shown that, w.h.p., $\|\bm\nu_k\| \le \sqrt{1.1} \|(\I - \U \U{}^\top)\xstar_k\|^2 \le \sqrt{1.1} \SE_F(\U,\Ustar) \|\xstar_k\|$, i.e., the noise is proportional to the error in $\U$ \footnote{For this particular bound, we could also have used the tighter bound  of $\sqrt{1.1} \SE_2(\U,\Ustar) \|\xstar_k\|$, however, since the later parts of our proof require use of $\SEF$, we use that here too to keep things consistent.}.

(2) Given a good estimate, $\bhat_k$, of $\tb_k$ (or actually of $\g_k$) and $\U$ of $\Ustar$, we  get an equally good estimate, $\xhat_k = \U \bhat_k$, of $\xstar_k$ and hence of the measurements' phases $\cb^*_\ik$ (lines 7, 8). We denote the phase estimates by $\hat\cb_\ik$.

(3) Finally, we obtain a new estimate of $\Ustar$ by using the estimates $\b_k$ and $\hat\cb_\ik$ and solving  a Least Squares (LS) problem; see line 10. Here $\b_k$ is the $k$-th column of $\B$ which is obtained from $\Bhat$ by QR decomposition as $\Bhat \qreq \R_B \B$ (line 9).
The output of the LS step, $\Uhat$, may not have orthonormal columns (line 10). So we use QR decomposition $\Uhat \qreq \U \R_U$ to get $\U$ with orthormal columns (line 11). 

For the standard PR step for recovering the $\tb_k$s, we can use any algorithm with order-optimal sample complexity: TWF \cite{twf} or RWF \cite{rwf} or AltMin with truncated spectral initialization (AltMin-TSI) \cite{altmin_irene_w}. 
For real-valued measurements, we assume RWF is used since it already has a guarantee for noisy standard PR and since it was used in our earlier work \cite{lrpr_it}.  RWF and TWF guarantees are only for the real-valued case.
For the complex-valued case, we assume that AltMin-TSI is used. In fact, we could also assume that this is used for both cases.

\subsection{AltMinLowRaP Guarantee for noise-free LRPR}
We have the following guarantee for AltMinLowRaP for solving the LRPR problem. 

\begin{theorem}[Real or Complex Gaussian noise-free measurements]
Consider Algorithm \ref{lrpr_th} and assume that Assumption \ref{right_incoh} (right singular vectors incoherence) holds. 
Set $T := C\log(1/\epsilon)$, $T_{\PR,t} = C (\log r + \log \kappa + c  t ) $, $\omega = 1.3 \sigmin^2/q$.  
If
\[
m_0 q \ge C \kappa^8 \mu^2 n r^3, \text{ and }   m_1 q \ge 
C \max\left(  \kappa^4 \mu^2 r \max( nr,  \log\log(1/\epsilon) ),  q \max(r , \log q,  \log n)  \right),
\]
then,
with probability (w.p.) at least $1-  \exp(-n) - n^{-10}$,
$$\SEF(\Ustar,\U^0) \le \deltainitfrob= c /\kappa^2,  \ \SEF(\Ustar,\U^t) \le 0.2^t   \deltainitfrob,$$
and  $\matdist(\Xstar,\Xhat^t) \le  \SEF(\Ustar,\U^t) \sigmax$ for all $t=0,1,\dots, T$.
Consequently, in $T = C\log(1/\epsilon)$ iterations,
	\[
	\SEF(\Ustar,\U^T) \le \epsilon, \text{ and }  \matdist(\Xstar,\Xhat^T) \le \epsilon \sigmax
	\]
The time complexity is $m q nr \log^2(1/\epsilon)$.
	\label{best_thm}
\end{theorem}
\begin{proof}  We provide the overall proof in Sec. \ref{proof_best_thm}; the lemmas used in the proof are proved in Appendix \ref{lemmas_proofs}.  \end{proof}

Observe that the lower bound on just $m_1$ of $\max(r,\log q, \log n)$ is small and redundant except when $q > nr$. However, it is {\em necessary} because, given an estimate of $\Ustar$, the recovery of the columns of $\tB$ is a decoupled $r$-dimensional standard noisy PR problem. For each of these $q$ problems to work accurately w.h.p., $m \gtrsim r$ is needed. To deal with the union bound over these $q$ decoupled problems and still guarantee  success with probability at least $1 - n^{-10}$, we also need $m \gtrsim \max(\log q,\log n)$.

From Theorem \ref{best_thm}, assuming $nr > \log \log(1/\epsilon)$,  the total number of samples per column, $m_\tot:=m_0 + 2T m_1$, needs to satisfy%
\[
m_\tot q  \gtrsim \kappa^4 \mu^2  \left( \kappa^4 n r^3 +  \max( nr^2,  qr,  q \log q, q \log n) \log (1/\epsilon) \right).
\]
 Theorem \ref{best_thm} provides an immediate corollary for AltMinLowRaP solving linear LRPR (compressive PCA) as well. It has much lower time complexity than that of the convex relaxation for this problem from \cite{lee2019neurips}. For accuracy $\epsilon < 1/r$, its sample complexity is also better.


\subsection{Discussion}\label{discuss}
The discussion below treats $\kappa, \mu$ as constants and ignores dependence on them.
%
\subsubsection{Discussion: Why the LRPR sample complexity cannot be improved any further}
The number of degrees of freedom in a rank-$r$ $n \times q$ matrix is $(n+q)r$. Thus, ignoring the log factor, our sample complexity for the AltMin iterations is sub-optimal by a factor of $r$, while that for initialization is sub-optimal by a factor of $r^2$. For non-convex solutions to the two related problems -- sparse PR (phaseless but global measurements) and LRMC (linear but non-global measurements) -- that have been extensively studied for nearly a decade, the best existing guarantees are sub-optimal: ignoring log factors, the best non-convex LRMC result \cite{rmc_gd} requires $mq$ to be   $\Omega((n+q) r^2 \log^2 (1/\epsilon) \log^2 n)$; while the best sparse PR results (including those for convex solutions to sparse PR) require $m$ to be  $\Omega (s^2 \log n \log(1/\epsilon))$ where $s$ is the sparsity level, e.g., \cite{cai,fastphase}.  Once initialized carefully, the LRPR complexity is similar to that of LRMC. 
However, for initialization, both the phaseless and the non-global measurements imply that we need two extra factors of $r$ compared to the optimal. 

The reason for these extra factors is as follows. Because of the non-global nature of our measurements, and those of LRMC, when bounding an error term,  we need to use the incoherence assumption to show that the distributions of its $mq$ summands are similar enough so that the concentration  bounds (matrix Bernstein for LRMC and sub-exponential Bernstein in our case) can be applied jointly over all the $mq$ summands; each summand is a function of one $\y_\ik,\a_\ik$. This introduces an extra factor of $r$ over the optimal  in both the initialization and the iterations' complexity, both for our problem and for LRMC. The second extra factor of $r$ in our initialization complexity is due to the phase being unknown. For LRMC, because the measurements are linear one can define a matrix $\X_\init$ for which $\E[\X_\init] = \Xstar$ and compute its top $r$ singular vectors as the initialization for $\Ustar$. However, we cannot define such a matrix for LRPR. The same is also true for standard PR and sparse PR: one cannot define a vector $\x_\init$ whose expected value equals, or is close to, the true signal $\xstar$. Instead, one needs to define a matrix $\M$ that is close to a matrix of the form $\xstar \xstar{}^\top + c \I$ (the top eigenvector of this matrix is proportional to $\xstar$). Similarly, in our case, we need to define a ``squared" matrix of the form $\M = \sum_\ik \y_\ik \a_\ik \a_\ik{}^\top$ (actually its truncated version, see line 3 of Algorithm \ref{lrpr_th}) that is close to $\Xstar \Xstar{}^\top + c \I$.


Like LRMC, for the linear-LRPR (compressive PCA) setting,
it {\em is} possible to define a matrix $\X_\init$ that is close to $\Xstar$. 
and compute $\U^0$ as its top $r$ singular vectors. We are studying this in ongoing work \cite{lrpr_gdmin} where we can show that, for linear-LRPR, $mq \gtrsim nr^2$ suffices even for initialization. This ongoing work also develops a projected gradient descent (GD) based solution.

For the iterations, the LRPR and linear-LRPR complexities match, whereas for initialization LRPR needs an extra factor $r$. A similar pattern is seen for non-convex noise-free sparse PR guarantees \cite{fastphase,cai} as well: these need $\Omega (s^2 \log n)$ samples for initialization, but only $\Omega (s\log n)$ samples for the iterations (which compares with the sample complexity of compressive sensing which is the linear version of sparse PR).
The intuitive reason for this is PR problems behave like linear problems in the vicinity of the true solution.





\subsubsection{Discussion: Improvement over our older work}
Our result from \cite{lrpr_it} needed $m_0 q \gtrsim nr^4$ and $m_1 q \gtrsim nr^4$. Here we have reduced the lower bound on $m_1 q$ (sample complexity of the AltMin iterations) to $nr^2$ and that of the initialization step to $nr^3$. We explain below the changes to our proof that enable this improvement.

The analysis of the AltMin iterations involves bounding two error terms that we call Term1 and Term2. Term1 is the ``linear" error term (this is the error that would occur even if our measurements were linear) while Term2 is the phase error term that only occurs in the phaseless setting. Bounding both these requires bounds on $\|\xhat_k - \xstar_k\|$ and $\|\Xhat - \Xstar\|_F$ (technically $\dist(.,.)$ and $\matdist(.,.)$ respectively).
Using the ideas described earlier while explaining the algorithm, given an estimate $\U$ of $\Ustar$, we can show that $\|\xhat_k - \xstar_k\| \lesssim \|(\I - \U \U^\top) \Ustar \tb_k\| \le \SE_2(\Ustar,\U) \|\xstar_k\|$ for each $k$.%


There are two main changes to our proof approach compared to \cite{lrpr_it}.
(1) The first is that we use $\SEF(\Ustar,\U)$ instead of $\SE_2(\Ustar,\U)$ to get a tighter bound on $\|\Xhat - \Xstar\|_F$ as follows.  In \cite{lrpr_it}, we used the bound on $\|\xhat_k - \xstar_k\|$ to conclude that $\|\Xhat - \Xstar\|_F = \sqrt{\sum_k \|\xhat_k - \xstar_k\|^2} \lesssim \SE_2(\Ustar,\U) \|\Xstar\|_F \le \SE_2(\Ustar,\U) \sqrt{r}\sigmax$.  
Instead, we now use $\|\xhat_k - \xstar_k\| \lesssim \|(\I - \U \U^\top) \Ustar \tb_k\|$ to show that $\|\Xhat - \Xstar\|_F \lesssim \sqrt{\sum_k \|(\I - \U \U^\top) \Ustar \tb_k\|^2} = \sqrt{\|(\I - \U \U^\top) \Ustar \tB\|_F^2} \le \SEF(\Ustar,\U) \|\tB\|_2 =   \SEF(\Ustar,\U) \sigmax$. Observe that we have eliminated a factor of $\sqrt{r}$ from our bound on $\|\Xhat - \Xstar\|_F$. This factor also gets eliminated from our bounds on Term1 and Term2. Even though $\SEF(\Ustar,\U)$ is larger than $\SE_2(\Ustar,\U)$, this does not matter since we are able to prove exponential decay of the bound on $\SEF(\Ustar,\U)$ too, as long as the initial estimate satisfies $\SEF(\Ustar,\U^0) \le c$.
(2)
The second change is where Cauchy-Schwarz is applied when bounding Term2. In \cite{lrpr_it}, we used Cauchy-Schwarz to first upper bound $\mathrm{Term2}^2$ by a product of two terms, each of which could be easily bounded. We then bounded each of the two product terms separately (both the expected value and concentration bounds were obtained separately). Instead, we now use Cauchy-Schwarz to only upper bound $\E[\mathrm{Term2}]$, but apply concentration bounds directly on $(\mathrm{Term2} - \E[\mathrm{Term2}])$. 

Because of these two changes, we can now show that if, $\SEF(\Ustar,\U^0) \le c$, and if $mq \gtrsim \frac{nr^2}{\min(\epsilon_1, \epsilon_2 )^2}$, and $m \gtrsim \max(r,\log q, \log n)$, then w.h.p., $\SEF(\Ustar,\U^{t+1} ) \lesssim \frac{ (\epsilon_1 + \epsilon_2 + \sqrt{\deltatfrob})  \deltatfrob  \sigmax }{\sigmin}$. Here $\deltatfrob$ is the upper bound on $\SEF(\Ustar,\U^t)$ and satisfies $\deltatfrob < \deltainitfrob$. Also, $\epsilon_1 \deltatfrob \sigmax$ is our high probability bound on Term1 and $(\epsilon_2 + \sqrt{\deltatfrob})  \deltatfrob \sigmax$ is our bound on Term2.
 Thus setting $\epsilon_1 =\epsilon_2 = c /\kappa$ and $  \deltainitfrob = c/\kappa^2$ suffices to prove that $\SEF(\Ustar,\U^{t+1}) \lesssim c \deltatfrob = c^{t+1} \deltainitfrob$. This translates to requiring $mq \gtrsim nr^2$.
In \cite{lrpr_it}, under the same lower bound on $mq$, we could only get $\SE_2(\Ustar,\U^{t+1}) \lesssim  \frac{(\epsilon_1 +\sqrt{  \epsilon_2 +   \delta_t})  \delta_t \sqrt{r} \sigmax}{\sigmin}$ where $\delta_t$ was the upper bound on $\SE_2(\Ustar,\U^t)$ and satisfied $\delta_t < \delta_0$. Consequently, \cite{lrpr_it} needed to set $\epsilon_2 = \delta_0 = c/r$ and $\epsilon_1 = 1/\sqrt{r}$. Plugging this into the lower bound on $mq$, this shows why \cite{lrpr_it} needs $mq \gtrsim nr^4$ for the iterations. 
%

Consider the initialization. 
As explained above, we only need $\deltainitfrob = c$ instead of $\delta_0 = c/r$ needed by \cite{lrpr_it}.
%
%
For the initialization step, we need to use the result of \cite{lrpr_it} (this cannot be improved further). It shows that, if $mq \gtrsim nr^2 / \delta_0^2$, then w.h.p., $\SE_2(\Ustar,\U^0) \le \delta_0$. To use this to guarantee  $\SEF(\Ustar,\U^0) \le \deltainitfrob = c < 1$, we need to set $\delta_0 =  \deltainitfrob/\sqrt{r}$. This is why we need $mq \gtrsim nr^3$ for the initialization while the result of \cite{lrpr_it} needed  $mq \gtrsim nr^4$.

\subsubsection{Discussion: Improving time complexity: Projected gradient descent (GD)}\label{discuss_timecomp}
The $\log^2(1/\epsilon)$ factor in the time complexity can be reduced to $\log(1/\epsilon)$ if we can develop a projected GD solution for LRPR.
Since, in our setting, the error terms are sums of sub-exponential random variables, in order to obtain useful concentration bounds, we need a projected GD approach for which (i)  we can obtain tight column-wise bounds, i.e., bound $\|\x_k - \xstar_k\|$ for each $k$ by $\delta_t \|\xstar_k\|$ where $\delta_t$ is the error level at iteration $t$; and  (ii)  the gradient expression is such that the maximum sub-exponential norm of each summand of $|\w^\top \nabla f(\X) \z|$ is small enough (is of order $\delta_t (r/q)$) for any unit vectors $\w,\z$. Here $\nabla f(\X)$ is the gradient of the squared loss cost function w.r.t. $\X$. 
For direct modifications of either of the projected GD approaches that have been studied for LRMC   \cite{fastmc,rmc_gd,rpca_gd,lafferty_lrmc}, it is not possible to simultaneously get (i) and (ii). For projected-GD on $\X$ \cite{fastmc,rmc_gd}, it is not possible to get (ii), and it is not clear how to get (i) either\footnote{Even if we can somehow prove (i), we will only get $|\w^\top \nabla f(\X) \z| \lesssim \delta_t \sqrt{r/q} \sigmax$.}.
For the alternating GD approaches  \cite{rpca_gd,lafferty_lrmc}, it is not possible to get (i). Moreover, these require a GD step size $\eta$ that is proportional to $1/r$, making the convergence $r$-times slower than geometric. Consequently the time complexity advantage is lost.
We thus need a novel approach. In ongoing work that will be on ArXiv soon  \cite{lrpr_gdmin}, we study the following approach and argue that it satisfies both (i) and (ii) and converges geometrically:  update $\U$ by one step of GD on $f(\U,\B)$ w.r.t. $\U$, followed by QR decomposition (to get a matrix with orthonormal columns); but for each new $\U$, update $\B$ by minimizing $f(\U,\B)$ over $\B$. The minimization over $\B$ decouples into $q$ $r$-dimensional standard PR problems which have negligible time complexity $mr q \log(1/\epsilon)$ (no dependence on $n$).


\section{Proving Theorem \ref{best_thm}}\label{proof_best_thm}

\subsection{Lemmas needed for analyzing the AltMin iterations} 



Let $\U \equiv \U^{t}, \Bhat \equiv \Bhat^{t}, \B \equiv \B^t, \Xhat \equiv \Xhat^t$.
By a simple modification to Lemma 3.9 in \cite{lrpr_it}, we obtain the following lemma.

\begin{lemma}[$\SE_F$ version of Lemma 3.9 in \cite{lrpr_it}]
	\label{key_lem_SEF}
	\begin{align}
\SEF(\Ustar,\U^{t+1}) \le
 \frac{ \mathrm{MainTerm} }{\sigma_{\min}(\Xstar \B^\top) - \mathrm{MainTerm}} \ \text{ where}
	\label{SEU_bnd_0}
	\end{align}
	\[
\mathrm{MainTerm}:=	\dfrac{ \max_{\W \in \mathcal{S}_W}|\mathrm{Term1}(\W)| + \max_{\W \in \mathcal{S}_W}|\mathrm{Term2}(\W)|}{ \min_{\W \in \mathcal{S}_W} \mathrm{Term3}(\W)},
	\]
	\begin{align*}
	\mathrm{Term1}(\W) & :=  \sum_{ik} \b_k{}^\top \W^\top \a_\ik \a_\ik{}^\top (\Xstar \B^\top \b_k - \xstar_k) \\ 
	\mathrm{Term2}(\W) & :=  \sum_{ik} (\bar{\cb^*_\ik}  \hat\cb_\ik - 1) (\xstar_k{}^\top \a_\ik) (\a_\ik{}^\top \W \b_k) , \\
	\mathrm{Term3}(\W) & :=  \sum_{ik}  |\a_\ik{}^\top \W \b_k |^2, \\
	\S_W & := \{\W \in \Re^{n\times r}:\ \|\W\|_F = 1 \},
	\end{align*}
	and $\cb^*_\ik, \hat\cb_\ik$ are the phases of $\a_\ik{}^\top \xstar_k$  and $\a_\ik{}^\top \xhat_k$.
\end{lemma}	
\begin{proof} See Appendix \ref{proof_key_lem}. \end{proof}

We bound the above terms in the next lemma.
\begin{lemma}
Assume that $\SEF(\U^t, \Ustar) \le \deltatfrob$ with $\deltatfrob  < \deltainitfrob = c/\kappa^2$. 
Then,
\ben
\item w.p. at least $1 - 2 \exp\left(nr (\log 17) -c \frac{mq \epsilon_1^2}{\kappa^2 \mu^2 r} \right) - \exp(\log q + r - c m )$,
\[
\max_{\W \in \S_{W}} |\mathrm{Term1}(\W)|  \le C m \epsilon_1 \deltatfrob \sigmax,
\]
\item  w.p. at least  $1-  2 \exp\left( nr \log(17)   - c \frac{m q\epsilon_2^2}{ \mu^2 \kappa  r} \right) -   \exp(\log q + r - c m )$,
\[
\max_{\W \in \S_W} |\mathrm{Term2}(\W)| \le   C m ( \epsilon_2  +   \sqrt{\deltatfrob} )  \deltatfrob \sigmax,
\]
\item w.p. at least	$1 - 2\exp\left(nr  (\log 17)  - c \frac{\epsilon_3^2 m q}{\mu^2 \kappa^2 r}\right) -   \exp\left(\log q  +  r  - c  m \right),$
	\begin{align*}
	& \min_{\W \in \S_W}  \mathrm{Term3}(\W)  \ge 0.5 (1-\epsilon_3 )m,
	\end{align*}
\item $\sigma_{\min}(\Xstar \B^\top) \ge  \sigmin$.

\een
\label{Terms_bnds}
\end{lemma}

\begin{proof} See Appendix \ref{proof_Terms_bnds_1}, \ref{proof_Terms_bnds_2}, \ref{proof_Terms_bnds_3}, \ref{proof_Terms_bnds_4}. \end{proof}
For bounding the above terms, we need bounds on $\dist(\g_k, \bhat_k)$, $\matdist(\G, \Bhat)$, and $\matdist(\Xstar, \Xhat)$, and we need to show incoherence of $\bhat_k$s. We do this next.
\begin{lemma}
\label{B_lemma}
Let  $\g_k := \U^\top  \Ustar \tb_k$.
Assume that $\SEF(\U^t, \Ustar) \le \deltatfrob$ with $\deltatfrob  < \deltainitfrob =  c/\kappa^2$. 
Then, w.p. at least $1 - \exp(\log q + r -c m)$,
\begin{align*}
& \dist\left( \g_k, \bhat_k \right) \le \dist\left( \xstar_k, \xhat_k \right) \le C \deltatfrob^2 \|\tb_k\|  \nonumber \\
&  \matdist(\G, \Bhat) \le \matdist(\Xstar, \Xhat) \le C \deltatfrob \sigmax
\end{align*}
and since $\b_k = \R_B^{-1} \bhat_k$, and $\sigma_{\min}(\R_B) = \sigma_{\min}(\Bhat) \ge \sigma_{\min}(\G) -  \matdist(\G, \Bhat)$,
\begin{align*}
\|\b_k\| \le  \frac{\dist(\bhat_k, \g_k) + \|\g_k\|}{0.95\sigmin - \matdist(\G, \Bhat)} \le  2 \kappa \mu \sqrt{r/q}.  
\end{align*}
\end{lemma}

\begin{proof} See Appendix \ref{proof_B_bnds_1} and \ref{proof_B_bnds_2}. \end{proof}

We prove these lemmas for the real measurements'  case in Appendix \ref{lemmas_proofs}.  This proof is simpler and illustrates the ideas clearly, and hence we give it first. The complex case proof needs only two main changes, we provide these in Appendix \ref{proofs_complex}.%

\subsubsection{Brief Proof Sketch}
Here we replace $\dist(\z_1, \z_2)$ by $\|\z_1 - \z_2\|$ and $\matdist(\Z_1,\Z_2)$ by $\| \Z_1 - \Z_2\|_F$ (valid w.l.o.g. as explained in the Notation section).
The proof of Lemma \ref{B_lemma} relies on the following ideas: (1) as explained earlier, given an estimate $\U$ of $\Span(\Ustar)$, the recovery of each $\tb_k$ is an $r$-dimensional noisy (standard) PR problem with noise proportional to $\|(\I - \U \U^\top ) \Ustar \tb_k\|$; also one can only recover $\g_k$ which is the rotated version of $\tb_k$; (2) using a result for noisy standard PR from \cite{rwf}, we can thus show that
$\| \g_k- \bhat_k\| \le C \|(\I - \U \U^\top ) \Ustar \tb_k\|$; (3) to bound $\|\G - \Bhat\|_F$, we use this bound and the following simple fact $\sum_k \|\M \tb_k\|^2 = \|\M \tB\|_F^2 \le \|\M\|_F^2 \|\tB\|^2 =\|\M\|_F^2  \sigmax^2$;
(4) to bound $\|\g_k - \bhat_k\|$ we use the above bound and $\|(\I - \U \U^\top ) \Ustar \tb_k\| \le \deltatfrob \|\tb_k\|$;
(5) the bounds on $\|\xhat_k - \xstar_k\|$ and $\|\Xhat - \Xhat\|_F$ follow similarly with using the fact that $\xstar_k = \U \g_k + (\I - \U \U^\top )\Ustar \tb_k$ and $\xhat_k = \U \bhat_k$.

The proof of Lemma \ref{Terms_bnds} uses Lemma \ref{B_lemma}, the sub-exponential Bernstein inequality, and the following  ideas.
Consider $\mathrm{Term1}$. Let $\p_k:= (\Xstar \B^\top \b_k - \xstar_k)$. By using the fact that $\B \B^\top = \I$, we can show that $\sum_k \p_k \b_k{}^\top = 0$ and thus $\E[\mathrm{Term1}] = 0$. Also, by using
$\Xhat \B^\top \b_k = \xhat_k$, Lemma \ref{B_lemma}, and $\Xhat (\B^\top\B - \I) = \bm{0}$, we can show that $\|\p_k\| \le \|\Xstar - \Xhat\|_F \|\b_k\| + \|\xhat_k - \xstar_k\| \lesssim \deltatfrob \sigmax \kappa \mu \sqrt{r/q}$, and that $\|[\p_1, \p_2, \dots \p_q]\|_F \lesssim \deltatfrob \sigmax$.
For  $\mathrm{Term2}$, in the real measurements'  case, we use the following idea: $ \sum_\ik \E[\indic_{ \cb^*_\ik \neq \hat\cb_\ik }  (\a_\ik{}^\top \xstar_k )^2] \lesssim \sum_\ik \frac{\dist(\xstar_k,\xhat_k)^3}{\|\xstar_k\|} \lesssim \deltatfrob^3 \sigmax^2$. The first inequality follows from the Term2 bound proof in \cite{lrpr_it}, while the second follows because of our use of $\SEF(\Ustar,\U)$ instead of $\SE_2(\Ustar,\U)$.%

\subsection{Proof of Theorem \ref{best_thm}}
The theorem is an immediate consequence of the following two claims.
We use our result from \cite{lrpr_it} for the initialization. For the AltMin iterations, Claim \ref{altmin_claim} below is an easy consequence of the lemmas given above.
\begin{claim}
[Claim 3.1 of ~\cite{lrpr_it}]
	\label{bounding_U}
	Pick a $\deltinit < 0.25$.
	Set the  rank estimation threshold $\omega = 1.3 \sigmin^2 / q$. 
	Then, w.p. at least
$	1 - 2\exp\left( n   -c \frac{\deltinit^2 mq}{ \kappa^4 r^2} \right)  - 2 \exp\left(-c \frac{\deltinit^2 mq}{ \kappa^4  \mu^2 r^2} \right)  $,
	 the rank is correctly estimated and
\[
	\SE_2(\U^0,\Ustar) \leq \deltinit.
\]
\end{claim}

\begin{claim}
\label{altmin_claim}
Assume that $\SEF(\Ustar,\U^\tau) \le \deltainitfrob= c/ \kappa^2$.
Then, w.p. at least
$1- T \left( \exp\left(nr -c \frac{mq}{\kappa^2 \mu^4 r} \right) - \exp(\log q + r - c m)    \right)$,
\begin{align}\label{SEFbnd}
\SEF(\Ustar, \U^\tau) \le \delta_{\tau,F}:= 0.2^\tau \deltainitfrob \ \forall \tau=0,1,\dots, T
\end{align}
\end{claim}
To apply this claim, we need $\SEF(\Ustar,\U^\tau) \le \deltainitfrob= c/ \kappa^2$. Since $\SEF(\U^0,\Ustar) \le \sqrt{r} \SE_2(\U^0,\Ustar)$, this means we have to apply Claim \ref{bounding_U} with $\deltinit = \deltainitfrob /\sqrt{r} =  c /\kappa^2 \sqrt{r}$.
Applying it this way, w.p. at least
$
1 - 2\exp\left( n   -c \frac{c mq}{ \kappa^8 \mu^2 r^3} \right),
$
 $\SEF(\U^0,\Ustar) \le \deltainitfrob= c/ \kappa^2$.
Thus, combining this with  Claim \ref{altmin_claim}, \eqref{SEFbnd} holds for all $\tau$ w.p.
\[
1 - 2\exp\left( n   -c \frac{c m_0 q}{ \kappa^8 \mu^2 r^3} \right)- T \left( \exp\left(nr -c \frac{m_1 q}{\kappa^4 \mu^2 r} \right) - \exp(\log q + r - c m)    \right).
\]
where $m_0$ denote the number measurements for initialization and $m_1$ for the AltMin iterations. Recall that $T = C \log (1/\epsilon)$. 
Consequently, if
$m_0q  \ge C \kappa^8 \mu^2 nr^3$,  $m_1 q \ge C \kappa^4 \mu^2 r \max( nr,  \log\log(1/\epsilon) )$, and $m_1 \ge C \max( r , \log q,  \log n)$,  then the $\SEF$ bounds of the theorem hold w.p. at least $1 - \exp(- c n) - \exp(-c nr) - n^{-10} > 1- n^{-10}$. 
The bound on $\|\Xhat - \Xstar\|_F$ follows using Lemma \ref{B_lemma}. This proves Theorem \ref{best_thm}.%

\begin{proof}[Proof of Claim \ref{altmin_claim}]
This follows by induction. The base case is $\tau=0$ and this follows by the assumption in the claim.
Assume that \eqref{SEFbnd} holds for $\tau = t$. 
We will be done if we can show that it holds for $\tau=t+1$.

Under the claim's assumption and induction assumption, $\SEF(\Ustar, \U) \le \deltatfrob := 0.2^t \deltainitfrob < \deltainitfrob = c / \kappa^2$. Thus, we can apply Lemma \ref{Terms_bnds}.
We have that: w.p. at least
$
1-  2 \exp\left(nr (\log 17) -c \frac{mq \min(\epsilon_1,\epsilon_2, \epsilon_3)^2}{\kappa^2 \mu^2 r} \right)  -   \exp(\log q + r - c m ) ,
$
\[
 \mathrm{MainTerm}
\le \frac{C m (\epsilon_1  + \epsilon_2  +   \sqrt{\deltatfrob} ) }{(1 - \epsilon_3) m }   \deltatfrob \sigmax
\le C \frac{(\epsilon_1  + \epsilon_2  +   \sqrt{\deltainitfrob} ) }{(1 - \epsilon_3) }   \deltatfrob \sigmax
\]
where the second inequality used $\deltatfrob \le \deltainitfrob$.
Set $\epsilon_1 = \epsilon_2 = 0.01/ C \kappa$, $\epsilon_3 = 0.01$ and $\deltainitfrob= 0.01/ C^2 \kappa^2$.
With this, $(\epsilon_1  + \epsilon_2  +   \sqrt{\deltatfrob} ) /(1-\epsilon_3) < 0.12/  (0.99 C \kappa ) < 0.14/(C \kappa ) $ and so
\[
 \mathrm{MainTerm}  \le 0.14 \deltatfrob \sigmin
\]
 By Lemma \ref{key_lem_SEF}, $\sigma_{\min}(\Xstar \B^\top ) \ge \sigmin$, and using $\deltatfrob < \deltainitfrob = 0.01/\kappa^2 < 0.01$ for bounding MainTerm in the denominator,
	\begin{align*}
\SEF(\Ustar,\U^{t+1})
& \le  \frac{ 0.14  \deltatfrob \sigmin  } { \sigmin -  0.14 \cdot 0.01\sigmin }
 \le  0.16  \deltatfrob  < 0.2 \deltatfrob:= \delta_{t+1,F}.
	\end{align*}
Using $\deltatfrob = 0.2^t \deltainitfrob$, we get that $\delta_{t+1,F} = 0.2^{t+1} \deltainitfrob$.
\end{proof}

\section{Stability guarantee for noisy LRPR} \label{noisy_lrpr}
Consider LRPR with noisy measurements,
\[
\y_{ik} = |\langle \a_{ik}, \xstar_k\rangle| + \vv_{ik}, \qquad i\in[m],~k\in[q]
\]
Here $\vv_{ik}$ is noise. Define the noise vector $\vv_k:=[\vv_{1k}, \vv_{2k}, \dots, \vv_{mk}]^\top$. Then $\y_k = |\A_k{}^\top \xstar_k| + \vv_k$.



\newcommand{\deltatvfrob}{\delta_{t_v,F}}

Let $\epsilon_v = 0.01/\kappa$.
Recall that $\deltatfrob = 0.2^t \deltainitfrob$. %
Let $t_v$ be the smallest integer $t$ for which
\begin{align}
\label{noiselev}
\max \left( \frac{1}{\epsilon_v}  \frac{\|\V\|_F}{\sqrt{m} \sigmax}  , \max_k \frac{\|\vv_k\|}{\sqrt{m} \|\xstar_k\|} \right) > \deltatvfrob:=  0.2^{t_v}   \deltainitfrob.
\end{align}

We can prove the following result.

\begin{theorem}[Stability to small additive noise]
\label{thm:main_res_stability}
Consider Algorithm \ref{lrpr_th}. Assume that Assumption \ref{right_incoh} holds and that $\|\vv_k\| \leq \epsilon_{snr} \|\xstar_k\|$ for all $k\in[q]$ with $\epsilon_{snr} =  c / r^2 \kappa^4$.
%
Under the sample complexity bounds of Theorem \ref{best_thm},
\[
\SEF(\U_0, \Ustar) \le  \deltainitfrob =   c/\kappa^2,
\]
and, for $t_v$ defined above, for all iterations $t$,
\[
\SEF(\Ustar,\U_{t+1}) \le 0.2^{ \min(t,t_v) } \deltainitfrob.
\]
\end{theorem}

\begin{proof} We prove this result in Appendix. \ref{sec:stability}. \end{proof}

Observe that Theorem \ref{thm:main_res_stability} needs the signal-to-noise ratio (SNR), $\|\vv_k\|^2/\|\xstar_k\|^2 \lesssim (1/r^2)$. This is needed to show that the initialization returns an estimate $\U_0$ with $\SE_F(\U_0, \Ustar) \le c$.  Assuming the initialization bound holds, the AltMin iterations do not use this bound. For the iterations, we prove that the error cannot reduce below the noise level; here noise level refers to the left hand side of \eqref{noiselev}.

 Theorem \ref{thm:main_res_stability} needs the SNR upper bound because we do not make any statistical assumptions on the noise, i.e., we try to get results similar to those of \cite{twf,rwf} for standard PR.
If we do impose the standard zero mean i.i.d. assumption on the $\vv_k$s, it should be possible to reduce the required bound on SNR to $1/{r}$. Moreover, if the initialization sample complexity is increased by a factor of $r$, then the SNR upper bound can be further reduced to a constant $c$.

We can compare Theorem \ref{thm:main_res_stability} with the result of  \cite{lee2019neurips} for linear-LRPR which is the only existing noisy case result for a similar problem. It assumes that the noise vectors are zero mean Gaussian and i.i.d. and independent of the data. This, and the higher sample complexity, along with the fact that it studies a convex optimization problem, is why it does not need an explicit SNR upper bound.

\section{Conclusions} \label{conclude}
This work studied the Alternating Minimization (AltMin) algorithm, AltMinLowRaP, for solving the low rank phase retrieval (LRPR) problem introduced in our previous work \cite{lrpr_it}. We provided a significantly improved sample complexity guarantee and discussed (based on comparison with existing work on related well-studied problems) why we believe the result cannot be improved further. We showed that, if the right singular vectors' incohence assumption holds, if the initialization sample complexity is at least $\Omega(\kappa^8 \mu^2 nr^3)$ and the AltMin iterations' sample complexity is at least $\Omega(\kappa^4 \mu^2 nr^2)$, then AltMinLowRaP converges geometrically w.h.p. Its time complexity is thus $mqnr \log^2(1/\epsilon)$.
A second contribution of this work is a proof of stability to small additive noise of the same algorithm under the same sample and time complexity assumptions.

\bibliographystyle{IEEEbib}
{\bibliography{../bib/tipnewpfmt_kfcsfullpap}}

\appendices \renewcommand\thetheorem{\Alph{section}.\arabic{theorem}}

\section{Proofs of the lemmas} \label{lemmas_proofs}

We provide the real measurements case proof here, and postpone the extra steps for the complex case to Appendix \ref{proofs_complex}. 
%
Without loss of generality, as done in past works on PR, e.g., \cite{pr_altmin,lrpr_it}, for simplicity, at some places, we assume that $\xstar_k$ is replaced by $\bar{z} \xstar_k$ where $z = \mathrm{phase}(\langle \xstar_k, \xhat_k \rangle)$. Here $\mathrm{phase}(z):= z/|z|$ and $\bar{z}$ denotes its complex conjugate. With this,  $\dist(\xstar_k, \xhat_k) = \|\xstar_k - \xhat_k\|$.  A similar replacement can be done for each column of $\Xhat$ as well so that $\matdist(\Xstar,\Xhat) = \|\Xstar - \Xhat\|_F$. This point is explained carefully in \cite{lrpr_it}, Appendix A-C.

%
Our proofs rely on the sub-exponential Bernstein inequality (Theorem 2.8.1 of \cite{versh_book}) and the fact that the product of two sub-Gaussians $X,Y$ with sub-Gaussian norms $K_X, K_Y$ is sub-exponential with sub-exponential norm $K_X K_Y$ (Lemma 2.7.7 \cite{versh_book}). We state this as the following concentration bound:
\begin{lemma}[Lemma 2.7.7 and Theorem 2.8.1 of \cite{versh_book}] 
	\label{prodsubg}
\label{productsubg}
	Let $X_{i}, Y_{i}$, $i=1,2,\dots, N$, be sub-Gaussian random variables with sub-Gaussian norms $K_{X_i}$ and $K_{Y_i}$ respectively and with $\E[X_iY_i] = 0$ and the pairs $\{X_i,Y_i\},i=1,2,\dots,N$ mutually independent for different $i$. Then
	\begin{align*}
		\Pr \left\lbrace   |\sum_{i}^N   X_i Y_i | \geq t \right\rbrace
	 \leq 2\exp{\left( -c \min{\left(\frac{t^2}{\sum_i K_{X_i}^2 K_{Y_i}^2 }, \frac{t}{\max_i{|K_{X_i}K_{Y_i} |}} \right) } \right)}.
	\end{align*}
%
\end{lemma}

\subsection{Proof of Lemma \ref{key_lem_SEF}} \label{proof_key_lem}
The proof is almost the same as that of   Lemma 3.9 of \cite{lrpr_it} given in Appendix B in \cite{lrpr_it}.
As shown in equation (21) in that proof,
\begin{equation}
\Uhat^{t+1} =  \Xstar  \B^\top  - \F . 
\label{Uhat_express}
\end{equation}
with $\F$ being the $n \times r$ matrix version of the $nr$-length vector $\F_{vec}$ defined in the proof. As shown in equations (24)-(27) of the proof,
\begin{equation}
\|\F_{vec}\| \le  \frac{\max_{\W \in \S_{\W}} |\mathrm{Term1}(\W)| + \max_{\W \in \S_{\W}} |\mathrm{Term2}(\W)|}{\min_{\W \in \S_{\W}} |\mathrm{Term3}(\W)|}:=   \mathrm{MainTerm}.
\label{Fvec_bnd}
\end{equation}
Consequently, using \eqref{Uhat_express},
\[
||(\I - \Ustar \Ustar{}^\top) \Uhat^{t+1}||_F = \|\F\|_F = \|\F_{vec}\| \le \mathrm{MainTerm}
\]
Since $ \Uhat^{t+1} \qreq \U^{t+1} \R_U^{t+1}$, this means that $$\SE_F(\Ustar,\U^{t+1}) \le \mathrm{MainTerm} \ ||\R_U^{-1}|| = \frac{\mathrm{MainTerm}}{ \sigma_{\min}(\R_U)}.$$
Using $\sigma_{\min}(\R_U) = \sigma_{\min}(\Uhat^{t+1})$ and \eqref{Uhat_express},   $\sigma_{\min}(\R_U) = \sigma_{\min}(\Xstar \B^\top  - \F) \ge \sigma_{\min}(\Xstar \B^\top ) - \|\F\| \ge  \sigma_{\min}(\Xstar \B^\top ) - \|\F\|_F =
\sigma_{\min}( \Xstar \B^\top ) - \|\F_{vec}\| \ge \sigma_{\min}( \Xstar \B^\top ) - \mathrm{MainTerm} $.  
Thus,
\begin{align}\label{eq:se_1}
\SE_F(\Ustar,\U_{t+1}) 
\le \frac{ \mathrm{MainTerm}}{  \sigma_{\min}(\Xstar \B^\top ) -  \mathrm{MainTerm}}
\end{align}

\subsection{Proof of Lemma \ref{Terms_bnds}: Term1 bound}\label{proof_Terms_bnds_1}
Let 
\begin{align*}
\p_k := \Xstar \B^\top  \b_k - \xstar_k , \text{ and }  \P  := [\p_1, \p_2, \dots, \p_q] =  \Xstar ( \B^\top  \B - \I).
\end{align*}
Recall that
$
	\mathrm{Term1}(\W)= \sum_{ik} \b_k{}^\top  \W{}^\top  \a_{ik} \a_{ik}{}^\top   \p_k.
$
The main change from \cite{lrpr_it} is the fact that we now obtain a tighter upper bound on $\|\p_k\|, \|\P\|_F$ that relies on our tighter bound on $\|\Xstar - \Xhat\|_F$ given in Lemma \ref{B_lemma}.

First, as done in the proof of Lemma 3.11 of \cite{lrpr_it} \footnote{$\E[\mathrm{Term1}(\W)] = \E[\trace(\mathrm{Term1}(\W)] =  \E[\trace(\W{}^\top  \sum_\ik \a_{ik} \a_{ik}{}^\top   \p_k \b_k{}^\top )] =  m \trace(\W{}^\top  \sum_k \p_k \b_k{}^\top )]$. Since $\B\B^\top  =I$, $\sum_k \p_k \b_k{}^\top  = \Xstar \B^\top  \B \B^\top  - \Xstar \B^\top  =  \bm{0}$.},
\[
\E[\mathrm{Term1}(\W)]=0.
\]
Next, we bound $\|\p_k\|$ and $\|\P\|_F$. 
Observe that $\Xhat = \U \Bhat$ with $\Bhat \qreq \R_B \B$. Since $\B \B^\top  = I$, this means that  
\[
\Xhat \B^\top  \b_k = \U \R_B \B \B^\top  b_k = \U \R_B \b_k = \U \hat\b_k = \xhat_k
\]
Thus, by Lemma \ref{B_lemma},
\begin{align}
\|\p_k\|& = \|\Xstar \B^\top  \b_k - \Xhat \B^\top  \b_k   + \xhat_k - \xstar_k\|  \nonumber \\
& \le \|\Xstar - \Xhat\| \ \|\B\| \ \|\b_k\| + \|\xstar_k - \xhat_k\|  \nonumber \\
& \le \|\Xstar - \Xhat\|_F   \|\b_k\| +  C \deltatfrob \|\xstar_k\|  \nonumber \\
& \le C\deltatfrob \sigmax \max(\|\bstar_k\|, \|\b_k\|) \le  C \deltatfrob \sigmax  ( \kappa  \mu  \sqrt{r/q} )
\end{align}
Using $\Xhat = \U \Bhat = \U \R_B \B$, $\B(\B^\top \B - \I) = 0$, $\|\B\B^\top  - \I\| \le 2$, and $\|\P \M\|_F \le \|\P\| \|\M\|_F$,
\begin{align}
\|\P\|_F
& = \|(\Xstar - \Xhat + \Xhat)(\B^\top \B - \I)\|_F \nonumber  \\
& =  \|(\Xstar - \Xhat)(\B^\top \B - \I)\|_F  \le 2 \|\Xstar - \Xhat\|_F  \le  C \deltatfrob \sigmax
\label{P_bnd_SEF}
\end{align}

Let $X_{ik} = \a_{ik}^\top  \W \b_k$ and $Y_{ik} =\a_{ik}^\top  \p_k$. Both are sub-Gaussian with sub-Gaussian norms $K_{X_{ik}} = \|\W\b_k\| \leq \|\W\|_F \|\b_k\| \le \|\b_k\|$, and $K_{Y_{ik}}\leq  \|\p_k \| $.
Applying Lemma \ref{productsubg} with $t = m \epsilon_1 \deltatfrob \sigmax$, and using the bounds on $\|\p_k\|$, $\|\P\|_F$,
	\begin{align*}
		\frac{t^2}{\sum_{ik} K_{X_{ik}}^2 K_{Y_{ik}}^2}
		& = \frac{m^2 \epsilon_1^2  \deltatfrob^2 \sigmax^2}{m \sum_k \|\b_k\|^2  \|\p_k \|^2}
= c \frac{m  \epsilon_1^2 \deltatfrob^2  \sigmax^2}{ \max_k \|\b_k\|^2 \|\P\|_F^2}
 \ge c \frac{mq  \epsilon_1^2}{\kappa^2 \mu^2 \ r}, \\
		\frac{t}{\max_\ik K_{X_{ik}} K_{Y_{ik}}}
		& = \frac{m \epsilon_1 \deltatfrob \sigmax}{\max_k \|\b_k\| \|\p_k \|}
  =  c \frac{mq \epsilon_1}{\kappa^2 \mu^2 r}
	\end{align*}
and, hence,
	\[
	\Pr \lbrace |\mathrm{Term1}(\W)| \le m \epsilon_1 \deltatfrob \sigmax \rbrace \ge 1 -  \exp\left( - c \frac{mq  \epsilon_1^2}{\kappa^2 \mu^2 \ r} \right)
	\]
Now we just need to extend our bound for all $\W \in \S_W$ by developing an epsilon-net argument.
This part is exactly analogous to the Term1 bound proof from \cite{lrpr_it} except we now use $|\mathrm{Term1}(\W)| \le m \epsilon_1 \deltatfrob \sigmax$ instead of
$|\mathrm{Term1}(\W)| \le m \epsilon_1 \delta_t \sqrt{r} \sigmax$.

Thus,  w.p. at least $1 - 2 \exp\left(nr (\log 17) -c \frac{mq \epsilon_1^2}{\kappa^2 \mu^2 r} \right) - \exp(\log q + r - c m )$,
$
\max_{\W \in \S_{W}} |\mathrm{Term1}(\W)|  \le C m \epsilon_1 \deltatfrob \sigmax.
$

\subsection{Proof of Lemma \ref{Terms_bnds}: Bound on Term2 for real measurements}\label{proof_Terms_bnds_2}
We have $\mathrm{Term2}(\W)  :=  \sum_{ik} (\bar{\cb^*_\ik}  \hat\cb_\ik - 1) ( \a_\ik{}^\top  \xstar_k) (\a_\ik{}^\top  \W \b_k) = \sum_{ik} (\hat\cb_\ik - {\cb^*_\ik})  | \a_\ik{}^\top  \xstar_k | (\a_\ik{}^\top  \W \b_k)$ and so
\[
|\mathrm{Term2}(\W) |
= |\sum_{ik} (\hat\cb_\ik - {\cb^*_\ik})  | \a_\ik{}^\top  \xstar_k | (\a_\ik{}^\top  \W \b_k) |
\le \sum_{ik} |\hat\cb_\ik - {\cb^*_\ik}|  | \a_\ik{}^\top  \xstar_k | |\a_\ik{}^\top  \W \b_k| := \mathrm{Term2abs}(\W)
\]

\subsubsection{Bounding $\E[\mathrm{Term2abs}(\W)]$}
Using the Cauchy-Schwarz inequality,
\begin{align}
\E[\mathrm{Term2abs}(\W)]^2 \le \sum_{ik} \E[(\hat\cb_\ik - {\cb^*_\ik})^2  (\a_\ik{}^\top  \xstar_k )^2 ] \ \sum_{ik} \E[|\a_\ik{}^\top  \W \b_k|^2]
\label{cauchy_1}
\end{align}
As also shown in the proof of Lemma 3.10 of \cite{lrpr_it}, using $\B \B^\top = \I$ and $\|\W\|_F = 1$,
\[
 \sum_{ik}\E[(\a_\ik{}^\top  \W \b_k)^2] = m \sum_k \|\W \b_k\|^2 = m \|\W \B\|_F^2   = m  
\]
Consider the first term in \eqref{cauchy_1} and let
\begin{align}
Q_\ik:= |(\hat\cb_\ik - \cb^*_\ik)   (\a_\ik{}^\top  \xstar_k )|^2 = \indic_{ \cb^*_\ik \neq \hat\cb_\ik } (\a_\ik{}^\top  \xstar_k )^2
\label{def_Q_ik}
\end{align}
denote its one summand. We get the second equality because, for real measurements, phase = sign.
We proved the following in the proof of Lemma 3.12 (Term2 bound proof) of  \cite{lrpr_it}. We state it as a lemma to keep things clear
\begin{lemma}[Real measurements \cite{lrpr_it}]
\[
\E[\Q_\ik] = \E[\indic_{ \cb^*_\ik \neq \hat\cb_\ik }  (\a_\ik{}^\top  \xstar_k )^2] \le C \frac{\dist(\xstar_k,\xhat_k)^3}{\|\xstar_k\|}.
\]
\end{lemma}
\begin{proof} See the few equations above equation (20) of \cite{lrpr_it} (in its proof of Term2 bound). \end{proof}
By the above lemma and Lemma \ref{B_lemma},
\begin{align}
  \sum_{ik} \E[ \indic_{ \cb^*_\ik \neq \hat\cb_\ik }  (\a_\ik{}^\top  \xstar_k )^2 ]
 & \le C m \sum_k \frac{\dist(\xstar_k,\xhat_k)^3}{\|\xstar_k\|}  \nonumber \\
& \le  C m \max_k \frac{\dist(\xstar_k,\xhat_k)}{\|\xstar_k\|} \matdist(\Xstar, \Xhat)^2
 =  C m \deltatfrob   \deltatfrob^2 \sigmax^2
\label{tight_Qik}
\end{align}
Notice how this bound is tighter by a factor of $r$ than that of (20) of \cite{lrpr_it}.
Combining the previous equations,
\begin{align}
\E[\mathrm{Term2abs}(\W)] \le  C m \sqrt{\deltatfrob} \cdot \deltatfrob  \sigmax  
\label{ETerm2}
\end{align}

\subsubsection{Concentration bounds}
Unlike in \cite{lrpr_it}, here we use Cauchy-Schwarz to only bound the expectation as done above, while applying the concentration bound, Lemma \ref{prodsubg}, directly to $|\mathrm{Term2abs}(\W) - \E[\mathrm{Term2abs}(\W)]|$. 
For real measurements,
\[
\mathrm{Term2abs}(\W) =  \sum_{ik}  \indic_{ \cb^*_\ik \neq \hat\cb_\ik }  | \a_\ik{}^\top  \xstar_k | |\a_\ik{}^\top  \W \b_k|
\]
We use the fact that \cite{rwf}\footnote{
let $\h_k= \xstar_k - \xhat_k$, $(\a_{ik}{}^\top  \xstar_k)(\a_{ik}{}^\top  \xhat_k) < 0$ implies  $(\a_{ik}{}^\top  \xstar_k)(\a_{ik}{}^\top  \xstar_k ) - (\a_{ik}{}^\top  \xstar_k)(\a_{ik}{}^\top  \h_k) < 0$ and hence $|\a_{ik}{}^\top  \xstar_k| < |\a_{ik}{}^\top  \h_k|$.}
$\cb^*_\ik \neq \hat{\cb}_\ik$ implies that $|\a_{ik}{}^\top  \xstar| < |\a_{ik}{}^\top  (\xstar_k - \xhat_k)|$. Consequently,
\[
\indic_{ \cb^*_\ik \neq \hat\cb_\ik } | \a_\ik{}^\top  \xstar_k |  \leq 2 \lvert \a_{ik}{}^\top  (\xstar_k - \xhat_k) \rvert
\]
Let $X_{ik} = \indic_{ \cb^*_\ik \neq \hat\cb_\ik } | \a_\ik{}^\top  \xstar_k |  $, 	and  $Y_{ik} = \a_{ik}{}^\top \W \b_k$. Using above,
 $K_{X_{ik}} \leq C {\|\xstar_k - \xhat_k\|}$. Also, $K_{Y_{ik}} = C\|\W\b_k\|$.

We will apply Lemma \ref{prodsubg} with $t= m \epsilon_2 \deltatfrob \sigmax$. Using $\|\W \B\|_F^2 = 1$ and Lemma \ref{B_lemma},
	\begin{align*}
\frac{t^2}{\sum_\ik K_{X_\ik}^2  K_{Y_\ik}^2}
	& =  \frac{m^2 \epsilon_2^2 \deltatfrob^2 \sigmax^2}{\sum_{ik} \|\W \b_k\|^2 \|\xstar_k - \xhat_k\|^2 }
  \geq \frac{m \epsilon_2^2\deltatfrob^2 \sigmax^2}{\max_k \|\xstar_k - \xhat_k\|^2 \sum_{k} \|\W \b_k\|^2  }
	 \geq c  \frac{m q\epsilon_2^2}{\mu  r},  \\
\frac{t}{\max_\ik K_{X_\ik}  K_{Y_\ik}} 	
	& \geq \frac{m \epsilon_2\deltatfrob \sigmax  }{  \max_{k} \|\W \b_k\| \max_k \|\xstar_k - \xhat_k \| }
	 \geq c \frac{m \epsilon_2\deltatfrob \sigmax  }{ ( \mu \kappa  \sqrt{r/q})  (\mu \deltatfrob \sigmax \sqrt{r/q} ) }
	 =  c\frac{mq \epsilon_2  }{ \mu^2 \kappa r }
	\end{align*}
Since $\min\left( \frac{m q\epsilon_2^2}{\mu  r},\frac{mq \epsilon_2  }{ \mu^2 \kappa r } \right)  = \frac{m q\epsilon_2^2}{ \mu^2 \kappa  r}$, applying Lemma \ref{prodsubg},
\begin{align*}
\Pr \left( | \mathrm{Term2abs}(\W) -\E[\mathrm{Term2abs}(\W)]|  \leq C m \epsilon_2  \deltatfrob \sigmax   \right)
	\geq  1-  2 \exp\left( - c \frac{m q\epsilon_2^2}{ \mu^2 \kappa  r} \right).
\end{align*}
Using \eqref{ETerm2},  this implies that, w.p. at least $1-  2 \exp\left( - c \frac{m q\epsilon_2^2}{ \mu^2 \kappa  r} \right)$,
\[
|\mathrm{Term2abs}(\W)|  \leq C m ( \epsilon_2  +   \sqrt{\deltatfrob} ) \deltatfrob \sigmax  .
\]%
Finally, we just need an epsilon-net argument to extend the above bound to all unit Frobenius norm matrices $\W$. This is pretty standard. We provide it here for completion. 

\subsubsection{Epsilon-net argument (pretty standard)}
By Lemma 5.2 of  \cite{vershynin} there exists a set (called epsilon-net), $\bar{\S}_W \subset \S_W$ so that, for any $\W \in \S_W$, there is a $\bar{\W} \in \bar{\S}_W$ such that $\|\bar{\W} - \W\|_F \leq \epsilon_{net}$ and $\lvert \bar{\S}_W \rvert \leq \left( 1+\frac{2}{\epsilon_{net}} \right)^{nr}$. By picking $\epsilon_{net} =1/8$ we have $\lvert \bar{\S}_W \rvert \leq (17)^{nr}$. Define $\Delta\W := \W - \bar{\W}$ so that $\|\Delta\W\|_F \leq \epsilon_{net} = \frac{1}{8}$. Using union bound for all entries in $\bar{\S}_W$,
	\begin{align}
	\Pr\left\{ \mathrm{Term2abs}(\bar{\W}) \leq  C m ( \epsilon_2  +   \sqrt{\deltatfrob} ) \deltatfrob \sigmax \ \text{for all $\bar\W \in \bar\S_W$}  \right\}
	& \geq 1-  2 \exp\left( nr \log(17)   - c \frac{m q\epsilon_2^2}{ \mu^2 \kappa  r} \right).
\label{Term2_epsnet}
	\end{align}
	Next we extend this for the entire hyper-sphere, $\S_W$. Define $\Gamma_W := \max_{\W \in \S_W}  \mathrm{Term2abs}(\W) $.
	Since $\frac{\Delta\W}{\|\Delta\W\|_F} \in \S_W$, 
using \eqref{Term2_epsnet},  for any $\W \in \S_W$,
w.p. at least $1-  2 \exp\left( nr \log(17)   - c \frac{m q\epsilon_2^2}{ \mu^2 \kappa  r} \right)$,
	\begin{align*}
	\mathrm{Term2abs}(\W) 
	& =   \sum_{ik} |\hat{\cb}_{ik} - \cb^*_{ik}|  | \a_{ik}{}^\top  \xstar_k | |\b_k{} ^\top (\bar{\W} + \Delta\W)^\top  \a_{ik}|  \\
	& \leq  C m ( \epsilon_2  +   \sqrt{\deltatfrob} )  \deltatfrob \sigmax  + \Gamma_W \|\Delta\W\|_F  \le C m ( \epsilon_2  +   \sqrt{\deltatfrob} )  \deltatfrob \sigmax + \Gamma_W \epsilon_{net}
	\end{align*}
where we used $\|\Delta\W\|_F \le \epsilon_{net}$ in the last bound.
Taking $\max$ over ${\W \in \S_W}$ on both sides,
\[
\Gamma_W = \max_{\W \in \S_W} \mathrm{Term2abs}(\W)  \le C m ( \epsilon_2  +   \sqrt{\deltatfrob} )  \deltatfrob \sigmax + \epsilon_{net} \Gamma_W
\]
and so  $ \Gamma_W \leq \frac{ C m ( \epsilon_2  +   \sqrt{\deltatfrob} )  \deltatfrob \sigmax }{1-\epsilon_{net}} = \frac{8}{7} C  m ( \epsilon_2  +   \sqrt{\deltatfrob} )  \deltatfrob \sigmax$ since we set $\epsilon_{net} = 1/8$.

Thus, using  $|\mathrm{Term2}(\W) | \le \mathrm{Term2abs}(\W)$,
w.p. at least  $1-  2 \exp\left( nr \log(17)   - c \frac{m q\epsilon_2^2}{ \mu^2 \kappa  r} \right) -   \exp(\log q + r - c m )$,
\[
\max_{\W \in \S_W} |\mathrm{Term2}(\W)| \le   C m ( \epsilon_2  +   \sqrt{\deltatfrob} )  \deltatfrob \sigmax.
\]

\subsection{Proof of Lemma \ref{Terms_bnds}: Lower bound on Term3}\label{proof_Terms_bnds_3}
This proof is the same as that of Lemma 3.10 of \cite{lrpr_it}. It uses the last claim of Lemma \ref{B_lemma} ($\b_k$s are $\hat\mu$-incoherent).
%

\subsection{Proof of Lemma \ref{Terms_bnds}: Lower bound on $\sigma_{\min}(\Xstar \B^\top )$}\label{proof_Terms_bnds_4}
$\sigma_{\min}(\Xstar \B^\top) \ge \sigma_{\min}(\Ustar \bSigma) \sigma_{\min}(\Bstar) \sigma_{\min}(\B^\top) = \sigmin$ since $\sigma_{\min}(\B) = \sigma_{\min}(\Bstar)=1$ (both are basis matrices).


\subsection{Proof of Lemma \ref{B_lemma}, part 1: Bound on $\|\g_k - \bhat_k\|$, $\|\G - \Bhat\|_F$ and $\|\Xstar - \Xhat\|_F$} \label{proof_B_bnds_1}

Recall that $\g_k = \U^\top \xstar_k$, and with this $\y_\ik =  |\a_\ik{}^\top  \U \g_k + \a_\ik{}^\top  (\I - \U \U^\top ) \xstar_k|$.
Thus, we can write
	\begin{align*}
	\y_\ik & = \lvert \a_{ik}{}^\top  \U g_k \rvert + \nu_{ik}, \text{ with }  \\
\nu_{ik} & =   |\a_\ik{}^\top  \U \g_k + \a_\ik{}^\top  (\I - \U \U^\top ) \xstar_k| - \lvert \a_{ik}{}^\top  \U \g_k \rvert
\end{align*}
Hence recovering $\g_k$ from $\y_\ik$s, $i \in [m]$, is a noisy $r$-dimensional standard PR problem with
\begin{align*}
\|\bnu_k\|^2 := \sum_{i} \nu_{ik}^2 \le \sum_{i}  |\a_\ik{}^\top  (\I - \U \U^\top ) \xstar_k|^2.
\end{align*}
Since both LHS and RHS are non-negative, we can take $\E[.]$ on both sides to conclude that
\[
\E[\|\bnu_k\|^2 ] \le \E[ |\a_\ik{}^\top  (\I - \U \U^\top ) \xstar_k|^2] =  m \|  (\I - \U \U^\top ) \xstar_k\|^2 = m \|  (\I - \U \U^\top ) \Ustar \tb_k\|^2
\]
Using above and applying Lemma \ref{prodsubg} with $t = m  \delta_b   \|  (\I - \U \U^\top ) \xstar_k\|^2$, $K_{X_{i}} = K_{Y_{i}} = \|  (\I - \U \U^\top ) \xstar_k\|$, and summing over $i=1,2,\dots,m$, we conclude that,
w.p. at least $1-\exp\left(  -c\delta_b ^2 m\right)$,
	\begin{align}
		{\|\bnu_k\|^2}  \leq m (1+\delta_b) \|  (\I - \U \U^\top ) \Ustar \tb_k\|^2  
\label{nuk_bnd}
	\end{align}
%
We estimate $\bhat_k$ by solving a standard PR problem using measurements $\y_\ik, i \in [m]$ with measurement vectors $(\U^\top  \a_\ik), i \in [m]$. We use RWF for this. 
 By Theorem 2 of \cite{rwf} for RWF,  if $m \ge Cr$, w.p. at least $1- \exp(r -c m)$,
	\begin{align*}
		& \dist\left( \g_k, \bhat_k \right) \leq C \frac{\|\bnu_k\|}{\sqrt{m}} + \rho^{T_{\PR,t}} \|\g_k\|,
	\end{align*}
	with $\rho<1$. 
By picking $T_{\PR,t}$ so that $\rho^{T_{\PR,t}} \le \deltatfrob / \sqrt{r}$, using $\|\g_k\| \le \|\tb_k\|$, using $(a+b)^2 \le 2 (a^2 + b^2)$, and finally using the bound on $\bnu_k$ from \eqref{nuk_bnd} with $\delta_b=0.1$,
we can conclude that,  
w.p. at least $1 - 2 \exp(r -c m)$, 
	\begin{align*}
		 \dist\left( \g_k, \bhat_k \right)^2 & \leq C \frac{\|\bnu_k\|^2}{m}  + C \frac{\deltatfrob^2}{r} \|\tb_k\|^2  \\
& \le C  \|  (\I - \U \U^\top ) \Ustar \tb_k\|^2  + C \frac{\deltatfrob^2}{r} \|\tb_k\|^2 \
 \le C \deltatfrob^2 \|\bstar_k\|^2.
	\end{align*}
By union bound, the above bound holds for all $k=1,2,\dots,q$, w.p. at least
\[
1- q\exp(r - c m) = 1 - \exp(\log q + r -c m).
\]
Since $\sum_k \|\M \tb_k\|^2 = \|\M \tB\|_F^2 \le \|\M\|_F^2 \|\tB\|^2 =\|\M\|_F^2  \sigmax^2$, we can use the second-last inequality from above to conclude that
	\begin{align*}
\matdist(\G, \Bhat)^2 & = \sum_k  \dist\left( \g_k, \bhat_k \right)^2  \\
&  \le C \sum_k \|  (\I - \U \U^\top ) \Ustar \tb_k\|^2  + C \frac{\deltatfrob^2}{r} \sum_k \|\bstar_k\|^2  \\
& \le  C  \|  (\I - \U \U^\top ) \Ustar \tB \|_F^2  + C \deltatfrob^2 (r \sigmax^2 ) / r \\
& \le  C  \deltatfrob^2 \sigmax^2  + C \deltatfrob^2 \sigmax^2  = C  \deltatfrob^2 \sigmax^2
	\end{align*}

Since $\xstar_k  - \xhat_k = \U \g_k + (\I - \U \U^\top ) \xstar_k  - \U \bhat_k =  \U (\g_k  - \bhat_k) + (\I - \U \U^\top ) \xstar_k $,
\begin{align*}
 \dist\left( \xstar_k, \xhat_k \right)^2 &  \le 	 \dist\left( \g_k, \bhat_k \right)^2  + \|(\I - \U \U^\top ) \Ustar \tb_k\|^2
\le C \deltatfrob^2 \|\tb_k\|^2.
\end{align*}
and, proceeding as before,
\[
\matdist(\Xstar, \Xhat)^2 \le C \sum_k  \dist\left( \g_k, \bhat_k \right)^2  + C \sum_k \|(\I - \U \U^\top ) \Ustar \tb_k\|^2  \le  C  \deltatfrob^2 \sigmax^2
\]
All the above claims hold w.p. at least $1 - \exp(\log q + r -c m)$.
%

\subsection{Proof of Lemma \ref{B_lemma}, part 2: Incoherence of columns of $\Bhat$}\label{proof_B_bnds_2}
	Recall that $\Bhat \qreq \R_B \B$ and so $\b_k = \R_B^{-1} \bhat_k$.
	Using the bound on $\|\bhat_k - \g_k\|$ from Lemma \ref{B_lemma}, and using $\|\g_k\| \le \|\tb_k\|$  and right singular vectors' incoherence (which implies that $\|\tb_k\|^2 \le \sigmax^2 \mu^2 r/q$),
	\begin{align*}
		\|\b_k\|  
		& \leq \| \R_B^{-1}\| \left( \dist(\bhat_k, \g_k) + \|\g_k\| \right)\\
		& \leq \frac{(1 + C\deltatfrob) \sigmax \mu \sqrt{r/q}}{\sigma_{\min}(\R_B)} 
	\end{align*}
Observe that $\sigma_{\min}(\R_B) = \sigma_{\min} (\Bhat )$. Using the bound on $\|\G - \Bhat\|_F$ from Lemma \ref{B_lemma}, w.p. at least $1 - \exp(\log q + r -c m)$, 
	\begin{align*}
		\sigma_{\min}(\R_B) = \sigma_{\min}(\Bhat) &\geq \sigma_{\min}\left(\G\right) - \|\G - \Bhat\| \\
		&	\geq \sigma_{\min}(\U^\top \U^*) \sigma_{\min}( \tB ) - \|\G - \Bhat\|_F \\
		& \ge  \sqrt{1- \SE_2^2(\U,\Ustar)} \sigmin  - C \deltatfrob \sigmax \\
		& \ge  \sqrt{1- \deltatfrob^2} \sigmin  -C \deltatfrob \sigmax.
	\end{align*}
Since we assumed $\deltatfrob \le c/ \kappa$,
$\sigma_{\min}(\R_B) \ge 0.9 \sigmin.$
Thus, letting $\hat\mu = C \kappa \mu$,
$
\|\b_k\| \le  \frac{C \sigmax \mu \sqrt{r/q} }{ 0.9 \sigmin } \le \hat\mu \sqrt{r/q} . 
$


\section{Extra proofs needed for the proof of Theorem \ref{best_thm} for complex measurements'  case} \label{proofs_complex}
As noted earlier, most of the steps of the proof are the same for the real and complex cases. The reason is we use concentration bounds from \cite{versh_book} and these apply (with minor changes to constants) for complex Gaussians as well.
The differences are in the bounding of Term2.
Secondly, the AltMin-TSI algorithm of \cite{altmin_irene_w} only comes with a noise-free case complex Gaussian measurements' guarantee, we need a bound for the noisy case. This replaces use of \cite[Theorem 2]{rwf} (which was proved only for the real case) in the proof of Lemma \ref{B_lemma}.

\subsection{Bounding Term2 for complex case}
Recall from the proof for the real case given in Appendix \ref{proof_Terms_bnds_2} that $|\mathrm{Term2}(\W)| \le \mathrm{Term2abs}(\W)$.

\subsubsection{Bounding $\E[\mathrm{Term2abs}(\W)]$}
By  Cauchy-Schwarz, \eqref{cauchy_1} holds and we still have $\E[\sum_{ik} |\a_\ik{}^\top  \W \b_k|^2] = m$. In this case, we need the following bound on $\E[Q_\ik ]$. 
\begin{lemma} Recall that $Q_\ik = |(\hat\cb_\ik - \cb^*_\ik)   (\a_\ik{}^\top  \xstar_k )|^2$. We have
\label{complex_EQ}
\[
\E[ |Q_\ik| ] \le C \frac{\dist^4(\xstar_k, \xhat_k)}{\|\xstar_k\|^2}.
\]
\end{lemma}
%
Using this lemma and proceeding as before,
\begin{align*}
\sum_{ik} \E[ | (\hat\cb_\ik - \cb^*_\ik)  (\a_\ik{}^\top  \xstar_k )|^2 ]
& \le C m \sum_k \frac{\dist(\xstar_k,\xhat_k)^4}{\|\xstar_k\|^2}
 =  C m \deltatfrob^2   \deltatfrob^2 \sigmax^2
\end{align*}
and thus
$
\E[\mathrm{Term2abs}(\W)] \le  C m  {\deltatfrob} \cdot \deltatfrob  \sigmax.
$

\subsubsection{Concentration bound}
Consider a fixed $\W$ first.
%
Let $\h_k = \xhat_{k} - \xstar_k$.
	Then, it is easy to see that $1- \bar{\cb^*_\ik} \hat\cb_\ik = 1-\mathrm{phase}\left( 1+\frac{\a_{ik}{}^\top\h_k}{\a_{ik}{}^\top\xstar_k}\right)$. By using Lemma A.7 of \cite{pr_altmin}, $\big|1-\mathrm{phase}\left(1+\frac{\a_{ik}{}^\top\h_k}{\a_{ik}{}^\top\xstar_k}\right)\big| \leq 2 \frac{|\a_{ik}{}^\top\h_k |}{|\a_{ik}{}^\top\xstar_k|}$.
Thus,
\[
 |  (\hat\cb_\ik  \cb^*_\ik -1 ) ( \a_\ik{}^\top \xstar_k) | \le 2 |\a_{ik}{}^\top\h_k |.
\]
 We can apply  Lemma \ref{prodsubg} with $X_\ik =  (\hat\cb_\ik  \cb^*_\ik -1 ) ( \a_\ik{}^\top \xstar_k)$, $K_{X,\ik} \le 2\|\h_k\|$, $Y_\ik = (\b_k^\top\W^\top\a_\ik)$, and $K_{Y,\ik} = \|\W \b_k\|$ exactly as in the real case.  After this the epsilon-net argument also follows as before.

\subsubsection{Proof of Lemma \ref{complex_EQ}} \label{new_complex_case}
Recall that $Q_\ik:= |(\hat\cb_\ik - \cb^*_\ik)   (\a_\ik{}^\top  \xstar_k )|^2$. Removing the indices for simplicity, we consider
$Q = |(\mathrm{phase}(\a^\top \xhat) -  \mathrm{phase}(\a^\top  \xstar)) (\a^\top \xstar)|^2 = | 1 - \mathrm{phase}( (\a^\top \xhat )(\a^\top  \xstar) | |(\a^\top \xstar)|^2$.

Let $\alpha = \frac{\langle\x,\xhat\rangle}{\|\x\| \|\xhat\|}$.
Define an orthonormal matrix $\O = [\o_1, \o_2, \O_{rest}]$ with $\o_1 = \xstar/\|\xstar\|$, $\o_2 = (\I - \o_1 \o_1{}^\top ) \xhat / \|(\I - \o_1 \o_1{}^\top ) \xhat\|$ and $\O_{rest}$ being any $n \times (n-2)$ matrix so that $\O$ is orthonormal. Since $\a^\top  \O$ has the same distribution as $\a$,
\[
Q = | 1 - \mathrm{phase}( (\a^\top \O^\top  \O \xhat )(\a^\top \O^\top  \O \xstar) ) | |(\a^\top \O^\top  \O \xstar)|^2
= | 1 - \mathrm{phase}( \alpha |\a_1|^2 +  \sqrt{1-\alpha^2} \a_1 \a_2 )  |\a_1|^2  \|\xstar\|^2
\]
We then bound its expected value by combining the two lemmas below.
\begin{lemma}
		\label{C first}
		Assume $a(1), a(2)$ are two independent standard complex Gaussian scalars and $0.8 \le \alpha \le 1$. Then we have
		\begin{align*}
		&  \E \left[\left( 1-  \mathrm{phase}\left(\alpha |a(1)|^2 + \sqrt{1-\alpha^2} a(2)\bar{a}(1)\right) \right) |a(1)|^2 \right]
		 \leq C \left(\frac{1-\alpha^2}{\alpha^2}\right)^2.
		\end{align*}
	\end{lemma}

\begin{lemma}
		\label{two_dists_connection}
		Consider two vectors $\x$ and $\xhat$. If $\dist(\x,\xhat) \le 0.5 \|\x\|$, then,
		$
		1-\frac{|\langle\x,\xhat\rangle|^2}{\|\x\|^2 \|\xhat\|^2} \leq 2 \frac{\dist(\x,\xhat)^2}{\|\x\|^2}.
		$
	\end{lemma}

\subsection{Proof of Lemmas \ref{C first} and \ref{two_dists_connection}}
We first provide some preliminary facts  needed for the first proof.
Let $\nu = \sqrt{\frac{1-\alpha^2}{\alpha^2}}$.	Since $\alpha \ge 0.8$, $\nu < 1$.
Observe that the $\mathrm{phase}$ term  can be expressed as
\begin{align*}
\mathrm{phase}\left(\alpha \big|a(1)\big|^2 + \sqrt{1-\alpha^2} a(2)\bar{a}(1)\right) =
\mathrm{phase}\left(1  + \nu \frac{a(2)}{a(1)} \right).
\end{align*}
%
Conditioned on $a(1)=w$, the term inside $\mathrm{phase}(.)$ is a complex Gaussian.
 Letting $a(2) = a_x + j a_y$ and $w = w_x + j w_y$, it equals
	\begin{align*}
	& Z:= \underbrace{1  + \frac{\nu}{|w|^2}\left( a_x w_x + a_y w_y\right)}_{X} + j \underbrace{ \frac{\nu}{|w|^2}\left( a_y w_x - a_x w_y\right) }_{Y}
	\end{align*}
	It is easy to see that
	$\sigma_X^2= \sigma_Y^2 = \frac{\nu^2}{ |w|^2}$, $\E\left[ X\right]=1$, $\E\left[Y \right]= 0$ and $X$ is uncorrelated with $Y$, ($\E[(X-1) (Y-0)] = 0$), so that $\rho=0$.
Thus, $Z$ is a non-zero mean complex Gaussian, with real and imaginary parts being independent and having the same variance $\sigma^2:=\frac{\nu^2}{ |w|^2}$ but different means: $\E[X]=1$ but $\E[Y]=0$. We will use a result from  \cite{dharmenvelope} that provides an expression for the PDF of the angle $\theta$ of such a complex Gaussian, i.e., for $\theta$, when we write $Z$ in polar form as $Z = R e^{j \theta}$.
%
	From \cite{dharmenvelope},
	\begin{align*}
	& f_\Theta (\theta)  = \frac{1}{2\pi} \exp\left\lbrace - \frac{1}{\sigma^2}\right\rbrace
	 \left\lbrace \sqrt{\frac{\pi \Omega_{X,Y}}{\Omega(\theta)} }\cos(\theta - \phi) \right. \\
	 &\left.\mathrm{erfc}\left( \frac{ -\sqrt{\Omega_{X,Y}}\cos(\theta - \phi)}{\sqrt{\Omega(\theta) }}\right)  \exp\left(  \frac{\Omega_{X,Y} \cos^2(\theta - \phi)}{\Omega(\theta)} \ \right)+1 \right\rbrace
	\end{align*}
	where $\Omega_{X,Y} = 1, \Omega(\theta) = 2 \sigma^2 , \frac{ \Omega_{X,Y}}{\Omega(\theta)}= \frac{1}{2 \sigma^2} , \cos\phi = \frac{1}{\sqrt{\Omega_{X,Y}}}= 1 \Rightarrow \phi = 0$.
	Hence we have
	\begin{align}
	& f_\Theta (\theta)  = \frac{1}{2\pi} \exp\left\lbrace - \frac{1}{\sigma^2}\right\rbrace \times   \nonumber \\
	& \left\lbrace \sqrt{\frac{\pi}{2 \sigma^2} }\cos(\theta ) \mathrm{erfc}\left( - \sqrt{\frac{ 1}{2 \sigma^2} }\cos(\theta)\right)  \exp\left(  \frac{\cos^2(\theta ) }{2\sigma^2}\ \right)+1 \right\rbrace   \nonumber  \\
& \stackrel{(1)}{\leq} \frac{1}{2\pi} \exp\left\lbrace - \frac{1}{\sigma^2}\right\rbrace \left\lbrace \sqrt{\frac{\pi}{2\sigma^2}} + 1 \right\rbrace .
	\label{ftheta_bnd}
	\end{align}
	where in $(1)$ we used the fact that $\mathrm{erfc}\left( - \sqrt{\frac{ 1}{2 \sigma^2} }\cos(\theta)\right)  \leq\exp\left(  \frac{\cos^2(\theta ) }{2\sigma^2}\ \right)$ along with $\cos(\theta ) \leq 1$.

We will use \eqref{ftheta_bnd} in the proofs below. Moreover, we will also frequently use the following:
for integers $n=1,2,3, 4, 5 \dots$
\begin{align}
\int_{\tau=0}^\infty  \tau^n \exp(- \tau^2/\nu^2) d\tau \le C \nu^{n+1}
\label{intbound}
\end{align}
where $C \leq 2$ for $n \leq 4$. This follows from the property of Gamma function that $\Gamma(z) = \int_{0}^{\infty} x^{z-1} \e^{-x} dx = (z-1)! $.

\begin{proof}[Proof of Lemma \ref{C first}]
		We need to bound
		\begin{align*}
		& \big| \E \left[\left( 1- \mathrm{phase}\left(\alpha \big|a(1)\big|^2  + \sqrt{1-\alpha^2} a(2)\bar{a}(1)\right) \right) |a(1)|^2\right] \big| \nonumber\\
		& =  \big| \E\left[ \underbrace{\left( 1- \mathrm{phase}\left(1  + \nu \frac{a(2)}{a(1)}\right) \right)  |a(1)|^2  }_{\mathrm{Trm1}} \right] \big|.
%
		\end{align*}
		First we bound $|\E[\mathrm{Trm1} |a(1) = w] |$.
		\begin{align*}
		 |\E\left[\mathrm{Trm1}| a(1) = w \right]|   =&|w|^2  | \E\left[ 1 -\mathrm{phase} \left(Z\right)\right] | \nonumber\\
		=& |w|^2 |  \E\left[ 1 - e^{j\theta}\right]  | \nonumber \\
		=& |w|^2 \big| \int_{0}^{2\pi} \left( 1 - e^{j\theta}\right) f_\Theta (\theta) d\theta \big|.
\label{eq:Expectation_1_minus_Z}
\end{align*}
	{\em	In the above $w$ is just a dummy variable that we are using for the conditional expectation as the known value for $a(1)$. It is completely different from matrix $\W \in \mathbb{C}^{n\times r}$ or its vectorized version $\w$ which was used previously.}
		
Since $f_\Theta (\theta) = f_\Theta (- \theta)$, $\int_{-\pi}^{\pi} \sin(\theta) f_\Theta (\theta) d\theta = 0$. Thus,%
		\begin{align}
		 \big| \int_{0}^{2\pi} \left( 1 - e^{j\theta}\right) f_\Theta (\theta) d\theta \big|
		& = \big|  \int_{-\pi}^{\pi} \left( 1 - \cos{\theta}\right) f_\Theta (\theta) d\theta \big| \nonumber\\
		& \leq \frac{1}{2\pi} \exp\left( - \frac{1}{\sigma^2}\right) \left( \sqrt{\frac{\pi}{2\sigma^2}} + 1 \right)  \int_{-\pi}^{\pi}  | 1 - \cos{\theta}|  d\theta \nonumber\\
		& \leq 2\exp\left( - \frac{1}{\sigma^2}\right) \left( \sqrt{\frac{\pi}{2\sigma^2}} + 1 \right)
         = 2\exp\left( - \frac{|w|^2}{\nu^2}\right) \left( \sqrt{\frac{\pi |w|^2}{2\nu^2}} + 1 \right)
		\end{align}
The first inequality used the upper bound on $f_\Theta(\theta)$ while the second used $|1 - \cos {\theta}| \le 2$. In the final equality, we substituted back $\sigma^2 = \frac{\nu^2}{|w|^2}$.
		Hence
		\begin{align*}
		 |\E[\mathrm{Trm1}]|  
		& \leq \sqrt{\frac{2\pi }{\nu^2}} \E\left[ |a(1) |^3 \exp\left( - \frac{ |a(1)|^2}{\nu^2}\right)\right]
				 + 2\E\left[ |a(1) |^2 \exp\left( - \frac{ |a(1)|^2}{\nu^2}\right)\right]   .
		\end{align*}
Since $a(1)$ is a standard complex Gaussian, $f_{|a(1)|} (x) = x e^{-\frac{x^2}{2}}, \forall\ x>0$. 
Using this,
		\begin{align*}
	    | \E[\mathrm{Trm1}] |
		&  \le \frac{2}{\sqrt{2\pi}}\int_{0}^{\infty} \left(\sqrt{\frac{2\pi }{\nu^2}} x^4 + 2 x^3 \right) e^{-x^2/\nu^2} e^{-x^2/2} dx\\
		&\leq \frac{2}{\sqrt{2\pi}}\int_{0}^{\infty} \left(\sqrt{\frac{2\pi }{\nu^2}} x^4 + 2 x^3 \right) e^{-x^2/\nu^2}  dx\\
		& \le C \nu^4.
		\end{align*}
		The second inequality used $e^{-x^2/2} \le 1$. The third one follows using (\ref{intbound}) with $n=4$ for the first term and $n=3$ for the second one.  
\end{proof}

\begin{proof}[Proof of Lemma \ref{two_dists_connection}]
 Define $\gamma^2 = 1-\frac{ |\langle\x,\xhat\rangle |^2}{\|\x\|^2 \|\xhat\|^2} $ and $\eta^2 = \min_{\theta \in [0,2\pi]} \|\x-e^{j\theta}\xhat\|^2$. Thus, we just need to show that $\gamma^2 \leq C \eta^2 /\|\x\|^2$ for $\eta \leq c \|\x\|$. To do this we can write
	\begin{align*}
	&\min_{\theta \in [0,2\pi]} \|\x-e^{j\theta}\xhat\|^2 = \|\x\|^2 + \|\xhat\|^2 - 2 \big|\langle\x,\xhat\rangle\big|.\\
&\Rightarrow	\eta^2 = \left(\|\x\| - \|\xhat\| \right)^2 + 2 \|\x\|~\|\xhat\| \left(1 -\sqrt{1-\gamma^2} \right)\\
	&\geq 2 \|\x\|~\|\xhat\| \left(1 -\sqrt{1-\gamma^2} \right)\\
	&\geq 2(1-c) \|\x\|^2 \left(1 -\sqrt{1-\gamma^2} \right),
	\end{align*}
	where in the last line we used the fact that $\|\xhat\| \geq \|\x\| - \|\x-\xhat\| \geq (1-c)\|\x\|$. This implies that
	\begin{align*}
&	\sqrt{1-\gamma^2} \geq 1- \frac{\eta^2}{2(1-c)\|\x\|^2}  \\
	& 1-\gamma^2 \geq 1+ \frac{\eta^4}{4(1-c)^2\|\x\|^4} -  \frac{\eta^2}{(1-c)\|\x\|^2}  \\		
	& \Rightarrow \gamma^2 \leq \frac{\eta^2}{(1-c)\|\x\|^2} \left(1 - \frac{\eta^2}{4(1-c)\|\x\|^2}\right)\\		
	 & \leq \frac{\eta^2}{(1-c)\|\x\|^2} ,
	\end{align*}
	where in the last inequality we used the fact that $0\leq\frac{\eta^2}{4(1-c)\|\x\|^2} \leq 1$.
\end{proof}

\subsection{Change to Proof of Lemma \ref{B_lemma}: Noisy PR result for AltMinTSI - modification of the result of \cite{altmin_irene_w}}

Since the PR problem is solved using AltMin-TSI from \cite{altmin_irene_w} for the complex case, we need the following result to analyze it. 
It follows by combining Theorem 2 of \cite{twf} and Theorem 3.1 of \cite{altmin_irene_w} with a minor change to deal with noise. 
	\begin{theorem}[Corollary 3.7 of \cite{altmin_irene_w}]
		\label{PRcomplexNoisy}
		Consider measurements of the form $\y_i = |\a_i{}^\top  \bm{g}^*| + \vv_i, \ i = 1, \cdots, m$, with $\vv$ satisfying $\frac{\|\vv \|}{\sqrt{m}} \leq c \|\xstar\|$. Here $\bm{g}^*$ is an $r$-length complex vector and $\a_i$ are i.i.d. complex standard Gaussian vectors of length $r$.
		Pick a $0 < \rho< 1$ and a $0< \rho_0<1$.
There exists a constant $C_0$ that depends on $\rho,\rho_0$, such that if $m> C_0(\rho,\rho_0) r$, then w.p. at least $1 - C \exp\left\{-c m\right\}$ for numerical constants C,c, the following holds after $T$ iterations:
		\begin{align*}
		 \dist(\hat\g^{T+1}, \bm{g}^*)  & \leq \rho^{T} \dist(\hat\g^0, \bm{g}^*) +  1.5\frac{1}{\sqrt{m}} \|\vv \|  \\
		               & \le \rho^{T}  \rho_0 \|\bm{g}^*\| +  3 \frac{1}{\sqrt{m}} \|\vv \|
		\end{align*}
By picking $T$ large enough, the first term above can be made smaller than the second; then, $\dist(\hat\g^{T+1}, \bm{g}^*) \le 6  \frac{1}{\sqrt{m}} \|\vv \|$.%
	\end{theorem}

{\bf Change to Proof of Lemma \ref{B_lemma}:} we use Theorem \ref{PRcomplexNoisy} to replace the application of Theorem 2 of \cite{rwf} in the proof.

\begin{proof}[Proof of Theorem \ref{PRcomplexNoisy}]
The initialization step of AltMin-TSI uses the truncated spectral initialization from \cite{twf}. Guarantees in \cite{twf} are proved for real-valued measurements. However, even with complex Gaussian measurements, there is no change to the  analysis of  truncated spectral initialization.
Thus we can use Theorem 2 of \cite{twf} with $t=0$ (only initialization part) to conclude that, w.p. $1 - C \exp(-\rho_0 m)$,
\[
\dist(\g^0, \bm{g}^*) \le \rho_0 \|\bm{g}^*\| + 2 \|\vv \| /\sqrt{m}.
\]
Consider iteration $t+1$.
Since $\hat\g^{t+1} = \A^{\dagger}\left(  \y \odot \mathrm{phase}(\A^\top  \hat\g^t) \right)$ (see Algorithm 1 of \cite{altmin_irene_w}),  where $\odot$ is the Hadamard product (.* operation in MATLAB), 
we have
\begin{align*}
	& \dist(\hat\g^{t+1}, \bm{g}^*)  =  \min_{\phi} \|\e^{j\phi} \bm{g}^* -\hat\g^{t+1} \|\\
	&=  \min_{\phi} \|\e^{j\phi} \bm{g}^* -(\A^\top )^{\dagger}\left(  (\y - \vv) \odot \mathrm{phase}(\A^\top  \hat\g^t) \right)\| \\
	&+ \| (\A^\top )^{\dagger}\left(  \vv \odot \mathrm{phase}(\A^\top  \hat\g^t) \right)\|\\
&	\leq  \rho \dist(\hat\g^{t}, \bm{g}^*) + \| \sqrt{m}  (\A^\top )^{\dagger}\| \  \|\frac{1}{\sqrt{m}} \vv \|.
	\end{align*}
The last inequality follows using Theorem 3.1 of \cite{altmin_irene_w}.
By the sub-Gaussian concentration bound from \cite[Theorem 4.6.1]{versh_book}, w.p. at least $1 - 2\exp(r -\epsilon m)$,
	\begin{align*}
	& \| \sqrt{m}  (\A^\top )^{\dagger}\|  \le     \| m (\A \A^\top  )^{-1}\|   \|   \A / \sqrt{m} \| \leq  \frac{1}{1-\epsilon} \sqrt{1+\epsilon} \le 1.5
	\end{align*}
	if we let $\epsilon = 0.1$. Thus,
\[
\dist(\hat\g^{t+1}, \bm{g}^*) \le  \rho \dist(\hat\g^{t}, \bm{g}^*)  +  1.5 \|\vv \| / \sqrt{m}.
\]
Using this and the initialization bound, $\dist(\hat\g^{t+1}, \bm{g}^*) \le \rho^t \rho_0 + 3 \|\vv \| / \sqrt{m}$.
\end{proof}

\section{Proof of Theorem \ref{thm:main_res_stability} for Noisy LRPR}
\label{sec:stability}

%

Most of the work is in modifying the proof of initialization. We also need a few simple changes to the rest of the proof.

\subsection{Initialization}
%
%
\begin{claim}
	\label{clm:noisy_init}
	Consider $\U^0$ to be output of the initialization step of Algorithm \ref{altmin_claim} with noisy measurements $\y_\ik = |\langle \a_\ik, \xstar_k\rangle| + \vv_\ik$ where $\vv_\ik$s are noise and with $\|\vv_k\| \leq \epsilon_{snr} \|\xstar_k\|$ with $$\epsilon_{snr} = \frac{\deltinit}{5r\kappa^2}.$$ Then, w.p. $1-\exp\left[Cn - c\deltinit^2 mq/r^2\mu^2\kappa^4\right]$ we have
	\[
	\SE_2(\U^0,\Ustar) \leq \deltinit.
	\]
\end{claim}

The proof is similar to the proof of \cite[Claim~3.1]{lrpr_it}.
Recall that $\y_\ik = |\langle \a_\ik, \xstar_k\rangle| + \vv_\ik$. Thus,
\[
|\y_\ik|^2 = |\langle \a_\ik, \xstar_k\rangle|^2 + |\vv_\ik|^2 + 2|\langle \a_\ik, \xstar_k\rangle \text{Re}(\vv_\ik) |.
\]
We also have
\[
\Y_U = \frac{1}{mq} \sum_\ik \y_\ik^2 \a_\ik\a_\ik{}^\top \indic_{ \left\{ \y_\ik^2 \leq \frac{C_Y}{mq}\sum_\ik \y_\ik^2 \right\}  }.
\]
Define
\begin{align*}
\Y_{U,clean} &:= \frac{1}{mq} \sum_\ik \big|\a_\ik{}^\top \xstar_k\big|^2 \a_\ik\a_\ik{}^\top \indic_{\{ |\y_\ik| \leq \sqrt{\frac{C_Y}{mq}\sum_\ik \y_\ik^2} \}},\\
\Y_{U,noise}&:= \frac{1}{mq} \sum_\ik \big|\vv_\ik\big|^2 \a_\ik\a_\ik{}^\top \indic_{\{ |\y_\ik| \leq \sqrt{\frac{C_Y}{mq}\sum_\ik \y_\ik^2} \}},\\
\Y_{U,cross}&:= \frac{2}{mq} \sum_\ik \vv_\ik \big|\a_\ik{}^\top \xstar_k\big| \a_\ik\a_\ik{}^\top \indic_{\{ |\y_\ik| \leq \sqrt{\frac{C_Y}{mq}\sum_\ik \y_\ik^2} \}}.
\end{align*}
Thus,
\[
\Y_U = \Y_{U,clean} + \Y_{U,noise} + \Y_{U,cross},
\]
Let
\[
\epsilon_1= \dfrac{\deltinit}{Cr\kappa^2}.
\]
Recall the definition of $\Y_-(\epsilon_1)$ and $\Y_+(\epsilon_1)$ in the proof of \cite[Claim~3.1]{lrpr_it}. We similarly define,
\[
\Y_-(\epsilon_1) := \frac{1}{mq}\sum_\ik\y_\ik^2\a_\ik\a_\ik{}^\top \indic_{\{ |\y_\ik|^2 \leq \sqrt{\frac{C_Y(1-5\epsilon_1)}{q} \|\Xstar\|_F} \}}
\]
and $\Y_+(\epsilon_1)$ by replacing $(1-5\epsilon_1)$ with $(1+5\epsilon_1)$.
  If we can show that $\| \Y_U - \E[\Y_{-}]\| \leq \frac{0.25 \deltinit\sigmin^2}{q}$ then the rest of the proof will be similar to the proof of \cite[Claim~3.1]{lrpr_it}. 
We have
\[
\| \Y_U - \E[\Y_{-}]\| \leq \|\Y_{U,clean}  - \E[\Y_{-}]\| + \|\Y_{U,noise}\| + \|\Y_{U,cross}\|.
\]
Starting with $\|\Y_{U,clean}  - \E[\Y_{-}]\| $, we bound each of the three terms on the RHS of the above inequality. Note that,
\begin{align*}
\sum_\ik \y_\ik^2 = \sum_\ik \big|\a_\ik{}^\top \xstar_k\big|^2 + \sum_\ik\big|\vv_\ik\big|^2 + 2\sum_\ik\big|\a_\ik{}^\top \xstar_k\big|\text{Re}(\vv_\ik).
\end{align*}
By  assumption, we have $\|\vv_k\|\leq \epsilon_{snr} \|\xstar_k\|$ and thus
\[
 \|\V\|_F \leq \epsilon_{snr}\|\Xstar\|_F
\]
In \cite[Lemma~3.6]{lrpr_it} it has been shown that $\sum_\ik \big|\a_\ik{}^\top \xstar_k\big|^2 \in m [(1\pm \epsilon_1)\|\Xstar\|_F^2]$ w.p. $1-\exp\left(-c\epsilon_1^2 mq/\mu^2\kappa^2\right)$. By using this, $\|\V\|_F \leq \epsilon_{snr}\|\Xstar\|_F$, and Cauchy-Schwarz inequality for the cross term,
\[
\sum_\ik \y_\ik^2 \leq m(1+\epsilon_1)\|\Xstar\|_F^2 + m\epsilon_{snr}^2 \|\Xstar\|_F^2 + 2 \sqrt{m(1+\epsilon_1)}\epsilon_{snr} \|\Xstar\|_F^2\leq m(1+5{\epsilon_1}) \|\Xstar\|_F^2,
\]
{\em Since $\epsilon_1$ is of the same order as $\epsilon_{snr}$ and hence in the above we replaced $\epsilon_{snr}$ by $\epsilon_1$ to simplify our bound.}
%
Similarly,
\[
\sum_\ik \y_\ik^2 \geq \sum_\ik \big|\a_\ik{}^\top \xstar_k\big|^2 - 2 \sum_\ik \big|\a_\ik{}^\top \xstar_k\big| \big|\vv_\ik\big| \geq m(1-\epsilon_1) \|\Xstar\|_F^2 - 2\sqrt{m(1+\epsilon_1)}\epsilon_{snr}\|\Xstar\|_F^2 \geq m(1-5\epsilon_1)\|\Xstar\|_F^2 
\]
Therefore, under assumption of $ \|\V\|_F \leq \epsilon_{snr} \|\Xstar\|_F$ w.p. $1-\exp\left(-c\epsilon_1^2 mq/\mu^2\kappa^2\right)$ we have
\[
\sum_\ik \y_\ik^2 \in m [(1\pm5{\epsilon_1}) \|\Xstar\|_F^2]
\]
\begin{itemize}
	\item In above, we have shown that, w.p. at least $1-2\exp(-c\epsilon_1^2mq/\mu^2\kappa^2)$,
	\[ \Y_{-}(5{\epsilon_1}) \preceq \Y_{U,clean}  \preceq \Y_{+}(5{\epsilon_1})
	\]
(this is an immediate consequence of the definitions and $\sum_\ik \y_\ik^2 \in m [(1\pm5{\epsilon_1}) \|\Xstar\|_F^2]$).
	\item In \cite[Lemma~3.7]{lrpr_it} we showed that
	\[
	\| \E[\Y_{+}] - \E[\Y_{-}] \| \leq \dfrac{45{\epsilon_1} \mu^2 \kappa^2  \|\Xstar\|_F^2 }{q}
	\]
	and assuming ${\epsilon_1} \leq 0.002$ we have $\min_k \beta_{1,k}^- (5{\epsilon_1}) \geq 1.5$. 
\item Moreover,
	\[
	\E\left[\Y_{-}(\epsilon_1) \right]= \sum_{k} \beta_{1,k}^- \xstar_k\xstar_k{}{}^\top  + \sum_{k}\beta_{2,k}^- \|\xstar_k\|^2\I_n,
	\]
	with $\beta_{1,k}^- \leq 1$ and $\beta_{2,k}^- \leq 2$. Therefore,
	\[
	\|\E\left[\Y_{-}(\epsilon_1) \right]\| \leq 3\|\Xstar\|_F^2.
	\]
	\item In \cite[Lemma~3.8]{lrpr_it} we show that w.p. at least $1-2\exp(n\ln 9-c\epsilon_2^2 mq)$,
	\[
	\|\Y_{-}-  \E[\Y_{-}] \| \leq \dfrac{1.5 \epsilon_2 \mu^2 \kappa^2 r \sigmax^2}{q}
	\]
	and the same bound for $\|\Y_{+}-  \E[\Y_{+}] \|$.
\end{itemize}
By using the above items we can show that
\[
\|\Y_{U,clean} - \E[\Y_{-}]\| \leq \dfrac{(45\epsilon_1 + 4.5 \epsilon_2)\mu^2\kappa^2r\sigmax^2}{q}.
\]
Therefore, we just need to bound $\|\Y_{U,noise}\| $ and $\|\Y_{U,cross}\|$. We present the following Lemmas to bound these two terms.
\begin{lemma}
	\label{lem:Y_U_noise}
	Under assumption of Theorem \ref{thm:main_res_stability},
	\[
	\|\Y_{U,noise}\| \leq \frac{C\epsilon_{snr}^2 \|\Xstar\|_F^2 }{q}
	\]
	holds with probability at least $1-2\exp(Cn -c \frac{mq}{\mu^2 \kappa^2 r})$.
\end{lemma}
\begin{lemma}
	\label{lem:Y_U_cross}
	Under assumption of Theorem \ref{thm:main_res_stability} and Lemma \ref{lem:Y_U_noise},
	\[
	\|\Y_{U,cross}\|  \leq \frac{C\epsilon_{snr} \|\Xstar\|_F^2 }{q}.
	\]
\end{lemma}
Therefore,
\begin{align*}
\| \Y_U - \E[\Y_{-}] \| &\leq \dfrac{(45\epsilon_1 + 4.5 \epsilon_2)\mu^2\kappa^2r\sigmax^2}{q}+ \dfrac{C(\epsilon_{snr} + \epsilon_{snr}^2)\|\Xstar\|_F^2}{q}\\
&\leq \dfrac{(45\epsilon_1 + 4.5 \epsilon_2)\mu^2\kappa^2r\sigmax^2}{q}+ \dfrac{C\epsilon_{snr} r\sigmax^2}{q}
\end{align*}
Recall that $\epsilon_1 = \tfrac{\deltinit}{Cr\kappa^2\mu^2}$ and $\epsilon_{snr} = \tfrac{\deltinit}{Cr\kappa^2}$. By setting
\[
 \epsilon_2 = \dfrac{\deltinit}{Cr\kappa^2\mu^2}, ~\epsilon_3 = c,
\]
we get
\[
\| \Y_U - \E[\Y_{-}] \| \leq \dfrac{0.25 \deltinit \sigmin^2}{q}.
\]
With this bound, one can apply the Davis-Kahan $\sin \Theta$ theorem exactly as done in the proof of \cite[Claim 3.1]{lrpr_it} to get the final result.

\begin{proof}[Proof of Lemma \ref{lem:Y_U_noise}]
	We have
	\begin{align*}
	\|\Y_{U,noise}\| = \max_{\w \in \mathbb{R}^{n}: \|\w\|=1} \w{}^\top  \Y_{U,noise} \w
	\end{align*}
	For a fixed $\w$, we have
	\begin{align*}
	\w{}^\top  \Y_{U,noise} \w &= \frac{1}{mq} \sum_\ik \big|\vv_\ik\big|^2\big|\a_\ik{}^\top \w\big|^2 \indic_{\{ |\y_\ik| \leq \sqrt{\frac{C_Y}{mq}\sum_\ik \y_\ik^2} \}}\\
	&\leq \frac{\epsilon_{snr}^2}{mq} \sum_\ik \|\xstar_k\|^2 \big|\a_\ik{}^\top \w\big|^2,
	\end{align*}
	where in the last inequality we used the assumption that $\|\vv_k\|\leq \epsilon_{snr} \|\xstar_k\|$.
	We can use the sub-exponential Bernstein inequality (note that $\|\xstar_k\|^2 \big|\a_\ik{}^\top \w\big|^2$ is a sub-exponential with norm less than $\|\xstar_k\|^2\leq \mu^2\kappa^2  \|\Xstar\|_F^2/q$ ) and show that
	\[
	\sum_\ik \|\xstar_k\|^2 \big|\a_\ik{}^\top \w\big|^2 \leq m(1+\epsilon_3)\|\Xstar\|_F^2
	\]
	with probability at least $1-2\exp(-c\epsilon_3^2 \frac{mq}{\mu^2 \kappa^2 r})$. Then, by using a standard epsilon net argument, we can show that
	\[
	\|\Y_{U,noise}\| \leq \frac{C\epsilon_{snr}^2 \|\Xstar\|_F^2 }{q}
	\]
	with probability at least $1-2\exp(Cn -c \epsilon_3^2\frac{mq}{\mu^2 \kappa^2 r})$.
	
\end{proof}

\begin{proof}[Proof of Lemma \ref{lem:Y_U_cross}]
	Similar to the previous proof,
	\begin{align*}
	\|\Y_{U,cross}\| &= \max_{\w \in \mathbb{R}^{n}: \|\w\|=1} \w{}^\top  \Y_{U,cross} \w\\
	&=\max_{\w \in \mathbb{R}^{n}: \|\w\|=1} \frac{2}{mq} \sum_\ik \text{Re}(\vv_\ik) \big|\a_\ik{}^\top \xstar_k\big|  \big|\a_\ik{}^\top \w\big|^2 \indic_{\{ |\y_\ik| \leq \sqrt{\frac{C_Y}{mq}\sum_\ik \y_\ik^2} \}}\\
	&\leq 2\left(\max_{\w \in \mathbb{R}^{n}: \|\w\|=1} \frac{1}{mq} \sum_\ik \big|\vv_\ik\big|^2   \big|\a_\ik{}^\top \w\big|^2 \indic_{\{ |\y_\ik| \leq \sqrt{\frac{C_Y}{mq}\sum_\ik \y_\ik^2} \}}\right)^{1/2}\\
	&\qquad \times \left(\max_{\w \in \mathbb{R}^{n}: \|\w\|=1} \frac{1}{mq} \sum_\ik  \big|\a_\ik{}^\top \xstar_k\big|^2  \big|\a_\ik{}^\top \w\big|^2 \indic_{\{ |\y_\ik| \leq \sqrt{\frac{C_Y}{mq}\sum_\ik \y_\ik^2} \}}\right)^{1/2}\\
	&= 2\|\Y_{U,clean}\|^{1/2}~\|\Y_{U,noise}\|^{1/2}\\
	&=2\left( \|\Y_{U,clean} - \E[\Y_{-}]\| + \|\E[\Y_{-}] \|  \right)^{1/2}\|\Y_{U,noise}\|^{1/2}.
	\end{align*}
	From the previous section we know that $\|\Y_{U,clean} - \E[\Y_{-}]\|  \leq \frac{0.25\deltinit \sigmin^2 }{q}$, and also $\|\E[\Y_{-}] \| \leq \max_k (\beta_{1,k}^- +\beta_{2,k}^-) \|\Xstar\|_F^2/q\leq 3\|\Xstar\|_F^2/q $. Therefore
	\begin{align*}
	\|\Y_{U,cross}\|
	&=2\left( \|\Y_{U,clean} - \E[\Y_{-}]\| + \|\E[\Y_{-}] \|  \right)^{1/2}\|\Y_{U,noise}\|^{1/2}\\
	&\leq \frac{3}{q} \|\Xstar\|_F\|\Y_{U,noise}\|^{1/2}\\
	& \leq \frac{C\epsilon_{snr} \|\Xstar\|_F^2 }{q}.
	\end{align*}
\end{proof}

\subsection{Updating $\B^t$: modify proof of Lemma \ref{B_lemma}}
Proceeding as in the proof of part 1 of Lemma \ref{B_lemma}, we now have
\begin{align}
\dist (\xstar_k,\xhat_k ) \approx \dist (\g_k,\bhat_k ) \leq C \frac{\|\vv_k \|}{\sqrt{m}}  + C \|(\I - \U \U {}^\top ) \Ustar \tb_k\|
\le C \max \left( \deltatfrob \|\xstar_k\|,   \frac{\|\vv_k \|}{\sqrt{m}} \right)
\label{gk_noisy}
\end{align}
and consequently,
\begin{align}
\matdist(\Xstar, \Xhat) \approx \matdist(\G, \Bhat) \leq C \frac{\|\V\|_F}{\sqrt{m}} + C\|(\I - \U \U {}^\top ) \Ustar \tB\|_F  \le C \max \left( \deltatfrob \sigmax, \frac{\|\V\|_F}{\sqrt{m}}   \right).
\label{G_noisy}
\end{align}
and
\begin{align}
\|\b_k\| \le \frac{\sigmax}{0.95 \sigmin - \|\G - \B\|_F}   \mu \sqrt{r/q}
\end{align}

\subsection{Updating $\U^t$}
For noisy measurements, the bound of Lemma \ref{key_lem_SEF} holds with the following change to MainTerm. We now have
\begin{align*}
\text{MainTerm}:= \dfrac{\max_{\W\in\S_W} \big|\text{Term1}(\W) \big| + \max_{\W\in\S_W} \big|\text{Term2}(\W) \big| + \max_{\W\in\S_W} \big|\text{TermNoise}(\W) \big| }{\min_{\W\in\S_W} \big|\text{Term3}(\W) \big|}
\end{align*}
where Term1 and Term2 are defined in Lemma \ref{key_lem_SEF} and
\begin{align*}
\text{TermNoise}(\W) = \sum_\ik  \hat\cb_\ik \vv_\ik \left(\a_\ik{}^\top \W\b_k\right).
\end{align*}
By using Cauchy Schwarz and the upper bound on $ \left(\sum_\ik\big|\a_\ik{}^\top \W\b_k\big|^2 \right)$ in the Term3 proof,
\begin{align*}
\text{TermNoise}(\W) \leq  \left(\sum_\ik \big|\vv_\ik\big|^2\right)^{1/2} ~ \left(\sum_\ik\big|\a_\ik{}^\top \W\b_k\big|^2 \right)^{1/2} \le ||\V||_F \sqrt{C m}
\end{align*}

Let $\epsilon_v = 0.01/\kappa$.
Recall that $\deltatfrob = 0.2^t \deltainitfrob$.  %
Let $t_v$ be the smallest integer $t$ for which
\[
\max \left( \frac{1}{\epsilon_v}  \frac{\|\V\|_F}{\sqrt{m} \sigmax}  , \max_k \frac{\|\vv_k\|}{\sqrt{m} \|\xstar_k\|} \right) \ge \deltatvfrob:=  0.2^{t_v}   \deltainitfrob
\]
Thus, for all $t < t_v$,  the first terms in \eqref{gk_noisy} and \eqref{G_noisy} are larger and, hence, the bounds of Lemma \ref{B_lemma} hold without change. Consequently, for these iterations, all the bounds of Lemma \ref{Terms_bnds} holds without change as well. Thus, for all $t < t_v$, with the stated probabilities,
\[
\max_{\W\in\S_W} \big|\text{Term1}(\W) \big| + \max_{\W\in\S_W} \big|\text{Term2}(\W) \big| \le (\epsilon_1 + \epsilon_2 + \sqrt{\deltatfrob} ) \deltatfrob \sigmax, 
\]
\begin{align*}
\max_{\W\in\S_W}\text{TermNoise}(\W) \leq   \|\V\|_F \sqrt{C m}  =  C m \left( \frac{1}{\sqrt{m}}\|\V\|_F \right) \le Cm \epsilon_v \deltatfrob \sigmax
\end{align*}
Hence for any $t < t_v$,
\[
\SEF(\Ustar,\U_{t+1}) \le \frac{  (\epsilon_1 + \epsilon_2 + \sqrt{\deltatfrob} + \epsilon_v) \deltatfrob \sigmax }{ 0.9 \sigmin - \text{Numerator}}  \le \frac{  (\epsilon_1 + \epsilon_2 + \sqrt{\deltainitfrob} + \epsilon_v) \deltatfrob \sigmax }{ 0.9 \sigmin - \text{Numerator}}
\]
Thus, setting $\epsilon_1 = \epsilon_2 = 0.01/\kappa$, $\deltainitfrob = c/\kappa^2$ as before, and using $\epsilon_v = 0.01/\kappa$, we can conclude that
\[
\SEF(\Ustar,\U_{t+1}) \le c \deltatfrob = c^{t+1} \deltainitfrob
\]

Next consider $t \ge t_v$. In this case, all we can guarantee is that
\begin{align}
\dist (\xstar,\xhat_k ) \approx \dist (\g_k,\bhat_k ) \leq C \deltatvfrob \|\xstar_k\|, \nonumber  \\
\matdist(\Xstar, \Xhat) \approx \matdist(\G, \Bhat) \leq C \deltatvfrob \sigmax, \nonumber  \\
\|\b_k\| \le \frac{\sigmax}{0.95 \sigmin - \|\G - \B\|_F}   \mu \sqrt{r/q} < 2 \kappa \mu \sqrt{r/q}
\end{align}
Notice that incoherence of $\b_k$s holds as before, because for proving this, $\deltatvfrob \le \deltainitfrob = c/\kappa^2$ suffices. Thus, even in this case, the lower bound on Term3 holds without change.
However, for Term1 and Term2, we can only prove the old bounds with  $\deltatfrob$ replaced by $\deltatvfrob$. The same is true for the bound on TermNoise.
Thus, for all $t \ge t_v$,
\[
\SEF(\Ustar,\U_{t+1}) \le \deltatvfrob = 0.2^{t_v} \deltainitfrob.
\]
In conclusion, for all times $t$, 
$
\SEF(\Ustar,\U_{t+1}) \le 0.2^{ \min(t,t_v) } \deltainitfrob.
$

\end{document}